\documentclass{CSML}

\def\dOi{11(4:18)2015}
\lmcsheading%
{\dOi}
{1--33}
{}
{}
{Jun.~24, 2014}
{Dec.~29, 2015}
{}

\usepackage{amssymb,amsmath,graphicx,epstopdf}
\usepackage{enumitem,moreenum}
\usepackage{framed}
\usepackage[all]{xy}
\usepackage{tikz}

\usetikzlibrary{positioning}

\ACMCCS{[{\bf Theory of computation}]: Computational complexity and cryptography---Complexity theory and logic; [{\bf Mathematics of
  computing}]: Discrete Mathematics---Graph Theory}

\usepackage{hyperref}
\theoremstyle{plain} %\crefname{satz}{Satz}{S\"atze}
\newtheorem{boundedwidththeorem}[thm]{Bounded width theorem}
\newtheorem{CSPdichotomy}[thm]{CSP dichotomy conjecture}

\newtheorem{algdichotomy}[thm]{Algebraic CSP dichotomy conjecture}

\newtheorem{finerCSPcomplexityconjectures}[thm]{Finer CSP complexity conjectures}
\newcommand{\CSP}{\operatorname{CSP}}
\newcommand{\up}[1]{\textup{#1}}
\DeclareMathOperator{\Pol}{Pol}
\DeclareMathOperator{\lvl}{lvl}
\newcommand{\G}{\mathbb{G}}
\newcommand{\A}{\mathbb{A}}
\newcommand{\B}{\mathbb{B}}
\renewcommand{\H}{\mathbb{H}} 

\newcommand{\Q}{\mathbb{Q}}
\newcommand{\hgt}[1]{\operatorname{hgt}(#1)}
\newcommand{\base}{\operatorname{base}}
\newcommand{\topp}{\operatorname{top}}

\begin{document}

\title[A finer reduction of constraint problems to digraphs]{A finer
  reduction of constraint problems to digraphs\rsuper*} 
\titlecomment{{\lsuper*}This article extends an earlier conference paper \cite{DBLP:conf/cp/BulinDJN13}}

\author[J.~Bul{\' i}n]{Jakub Bul{\' i}n\rsuper a}	%required
\address{{\lsuper a}Department of Mathematics, University of
	Colorado Boulder, USA}	%required
\email{jakub.bulin@colorado.edu}  %optional
\thanks{{\lsuper a}The first author was supported by the grant projects GA {\v C}R
13-01832S, GA UK 558313, SVV-2014-260107, M{\v S}MT {\v C}R 7AMB13PL013 and the Austrian Science Fund project P24285.}	%optional

\author[D.~Deli\'c]{Dejan Deli\'c\rsuper b}	%optional
\address{{\lsuper b}Department of Mathematics, Ryerson University,  Canada}	%optional
\email{ddelic@ryerson.ca}  %optional
\thanks{{\lsuper b}The second author gratefully acknowledges support by the Natural Sciences and
Engineering Research Council of Canada in the form of a Discovery Grant.}	%optional

\author[M.~Jackson]{Marcel Jackson\rsuper c}	%optional
\address{{\lsuper{c,d}}Department of Mathematics and Statistics, La Trobe University, Australia}	%optional
\email{m.g.jackson@latrobe.edu.au, toddniven@gmail.com}
%\urladdr{name3@url3\quad\rm{(optionally, a web-page can be specified)}}  %optional
\thanks{{\lsuper{c,d}}The
third fourth authors were supported by ARC Discovery Project DP1094578,
while the third author was also supported by Future Fellowship FT120100666.}	%optional

\author[T.~Niven]{Todd Niven\rsuper d}	%optional
\address{\vspace{-18 pt}}	%optional
%\email{toddniven@gmail.com}
%\thanks{The
%fourth author was supported by ARC Discovery Project DP1094578}
%% etc.

%% required for running head on odd and even pages, use suitable
%% abbreviations in case of long titles and many authors:

%%%%%%%%%%%%%%%%%%%%%%%%%%%%%%%%%%%%%%%%%%%%%%%%%%%%%%%%%%%%%%%%%%%%%%%%%%%

%% the abstract has to PRECEDE the command \maketitle:
%% be sure not to issue the \maketitle command twice!

\begin{abstract}
\noindent It is well known that the constraint satisfaction problem over a general relational 
structure $\mathbb{A}$ is polynomial time equivalent to the constraint problem over some associated digraph.   We present a
variant of this construction and show that the corresponding constraint satisfaction problem is logspace equivalent to that over $\mathbb{A}$.  Moreover, we show that almost all of the commonly encountered polymorphism properties are held equivalently on the $\mathbb{A}$ and the constructed digraph.  
As a consequence, the Algebraic CSP dichotomy conjecture as well as 
the conjectures characterizing CSPs solvable in logspace and in nondeterministic 
logspace are equivalent to their restriction to digraphs.
\end{abstract}

\maketitle

%% start the paper here:

%%%%%%%fix all this
\section*{Introduction}
\noindent A fundamental problem in constraint programming is to understand the
computational complexity of constraint satisfaction problems
(CSPs). While it is well known that constraint satisfaction
problems can be NP-complete in general, there are many subclasses of problems for
which there are efficient solving methods. One way to restrict the
instances is to only allow a fixed set of constraint relations, often
referred to as a \emph{constraint language}~\cite{b-j-k} or
\emph{fixed template}. Classifying the computational complexity of
fixed template CSPs has been a major focus in the theoretical study of
constraint satisfaction. In particular it is of interest to know which
templates produce polynomial time solvable problems to help provide
more efficient solution techniques.

The study of fixed template CSPs dates back to the 1970s with the
work of Montanari~\cite{Montanari} and Schaefer~\cite{sch}. A standout
result from this era is Schaefer's dichotomy for \emph{boolean} CSPs (i.e., arising from constraint languages over 2-element domains). The decision problems for fixed template CSPs over finite domains belong to the class NP, and Schaefer showed that in the boolean case, a constraint language is either solvable in
polynomial time or NP-complete. Dichotomies cannot be expected for
decision problems in general, since (under the assumption that
P$\neq$NP) there are many problems in NP that are neither solvable in
polynomial time, nor NP-complete \cite{lad}. Another important
dichotomy was proved by Hell and Ne{\v s}et{\v
  r}il~\cite{helnes:1}. They showed that if a fixed template is a
finite simple graph (the vertices make up the domain and the edge relation
is the only allowed constraint), then the corresponding CSP is
either polynomial time solvable or NP-complete. The decision problem
for a graph constraint language can be rephrased as a graph homomorphism
problem (a graph homomorphism is a function from the vertices of one
graph to another such that the edges are preserved).  Specifically,
given a fixed graph $\mathcal H$ (the constraint language), an
instance is a graph $\mathcal G$ together with the question ``Is there
a graph homomorphism from $\mathcal G$ to $\mathcal H$?''. In this
sense, $3$-colorability corresponds to $\mathcal H$ being the complete
graph on $3$ vertices. The notion of graph homomorphism problems
naturally extends to directed graph (digraph) homomorphism problems
and to relational structure homomorphism problems.

These early examples of dichotomies, by Schaefer, Hell and Ne{\v
  s}et{\v r}il, form the basis of a larger project of classifying the
complexity of fixed template CSPs.  Of particular importance in this
project is to prove the so-called \emph{CSP dichotomy conjecture} of
Feder and Vardi~\cite{fedvar} dating back to 1993. It states that the
CSPs related to a fixed constraint language over a finite domain are
either polynomial time solvable or NP-complete.  To date this
conjecture remains unanswered, but it has driven major advances in the
study of CSPs.
 
One such advance is the algebraic connection revealed in the work of Jeavons,
Cohen and Gyssens~\cite{JCG97} and later refined by Bulatov, Jeavons
and Krokhin~\cite{b-j-k}. This connection associates with each finite
domain constraint language~$\mathbb{A}$ a finite algebraic structure, the so-called
\emph{algebra of polymorphisms}. The properties of this algebraic structure are 
deeply linked with the computational complexity of the constraint
language. In particular, for a fixed core constraint language~$\mathbb{A}$,
if the algebra of polymorphisms of $\mathbb A$ does not satisfy a certain natural 
property, sometimes called being \emph{Taylor}, then the CSP restricted to the constraint language given
by~$\mathbb{A}$ is NP-complete. Bulatov, Jeavons and
Krokhin~\cite{b-j-k} go on to conjecture that all constraint
languages (over finite domains) whose algebras of polymorphisms are Taylor 
determine polynomial time CSPs (a
stronger form of the CSP dichotomy conjecture, since it describes
where the split between polynomial time and NP-completeness
lies). This conjecture is often referred to as the \emph{Algebraic CSP dichotomy conjecture}.  
  
Many important results have been built upon
this algebraic connection.  Bulatov~\cite{bul3} extended
Schaefer's~\cite{sch} result on 2-element domains to prove the CSP
dichotomy conjecture for 3-element domains. Barto, Kozik and
Niven~\cite{b-k-n} extended Hell and Ne{\v s}et{\v r}il's
result~\cite{helnes:1} on simple graphs to constraint languages
consisting of a finite digraph with no sources and no sinks. Barto and
Kozik~\cite{BW_journal_version} gave a complete algebraic description of the
constraint languages over finite domains that are solvable by local
consistency methods (these problems are said to be of \emph{bounded
  width}) and as a consequence it is decidable to determine whether a
constraint language can be solved by such methods. 

The algebraic approach was also succesfully applied to study finer complexity 
classification of CSPs. Larose and Tesson \cite{lartes} conjectured a natural algebraic 
characterization of templates giving rise to CSPs solvable in logspace (L) 
and in nondeterministic logspace (NL). In both cases they established the hardness 
part of the conjecture.

In their seminal paper, Feder and Vardi~\cite{fedvar} not only
conjectured a P vs. NP-complete dichotomy, they also reduced the problem of proving the
dichotomy conjecture to the particular case of digraph homomorphism
problems, and even to digraph homomorphism problems where the digraph
is balanced (here balanced means that its vertices can be partitioned
into levels).  Specifically, for every template $\mathbb{A}$ (a finite
relational structure of finite type) there is a balanced digraph 
$\mathcal{D}(\mathbb A)$ such that the CSP over $\mathbb A$ is
polynomial time equivalent to that over $\mathcal{D}(\mathbb{A})$.

In this paper we present a variant of such a construction and prove that
(under our construction) CSP over $\mathcal{D}(\mathbb{A})$ is \emph{logspace} equivalent to 
CSP over $\mathbb A$ and that the algebra of polymorphisms of 
the digraph $\mathcal{D}(\mathbb{A})$ retains almost all relevant properties. For example,
$\mathcal{D}(\mathbb{A})$ has bounded width, if and only if $\mathbb A$ does.
In particular, it follows that the Algebraic CSP dichotomy conjecture, 
the conjectures characterizing CSPs in L and NL as well as other open questions reduce 
to the case of digraphs.

In a conference version of this article \cite{DBLP:conf/cp/BulinDJN13}, the authors showed that the Algebraic CSP dichotomy conjecture is equivalent to its restriction to the case of digraphs.  This was established by showing that our construction preserves
a particular kind of algebraic property, namely existence of a 
\emph{weak near-unanimity} polymorphism.

\subsection*{Organization of the paper}

In Section 1 we present the main results of this paper. Section 2 introduces
our notation and the necessary notions concerning relational structures, digraphs and the algebraic approach to the CSP.
In Section 3 we describe the construction of $\mathcal D(\mathbb A)$.
Sections 4 and 5 are devoted to proving that the construction preserves cores 
and a large part of the equational properties satisfied by the algebra of polymorphisms.
Section 6 contains the logspace reduction of $\mathrm{CSP}(\mathcal D(\mathbb A))$ to
$\mathrm{CSP}(\mathbb A)$.
In Section 7 we discuss a few applications of our result and related open problems.

\section{The main results}\label{sec:mainresults}
\noindent In general, fixed template CSPs can be modelled as relational
structure homomorphism problems~\cite{fedvar}. For detailed 
definitions of relational structures, homomorphisms and other notions used in 
this section, see Section~\ref{sec:definitions}. 

Let $\mathbb{A}$ be a finite structure with signature~$\mathcal{R}$
(the fixed template).  Then the \emph{constraint satisfaction problem
  for} $\mathbb{A}$ is the following decision problem.

\medskip
\noindent \textbf{Constraint satisfaction problem for $\mathbb A$.}\\
\fbox{\parbox{0.67\textwidth}{ $\mathbf{CSP}(\mathbb A)$ \hrule
    \medskip
    INSTANCE: A finite $\mathcal{R}$-structure $\mathbb{X}$.\\
    QUESTION: Is there a homomorphism from $\mathbb{X}$ to
    $\mathbb{A}$?}}  \medskip

\noindent The dichotomy conjecture~\cite{fedvar} can be stated as
follows:
\begin{CSPdichotomy}
  Let $\mathbb{A}$ be a finite relational structure.  
  Then\linebreak
  $\CSP(\mathbb{A})$ is solvable in polynomial time or NP-complete.
\end{CSPdichotomy}
\noindent Every finite relational structure $\mathbb A$ has a unique
\emph{core} substructure $\mathbb A'$ (see Section~\ref{sec:rel
  struct} for the precise definition) such that $\CSP(\mathbb{A})$ and
$\CSP(\mathbb{A}')$ are identical problems, i.e., the ``yes'' and
``no'' instances are precisely the same. The algebraic dichotomy
conjecture~\cite{b-j-k} is the following:
\begin{algdichotomy}
  Let $\mathbb{A}$ be a finite relational structure that is a core. If
  the algebra of polymorphisms of 
  $\mathbb{A}$ is Taylor, then $\CSP(\mathbb{A})$ is
  solvable in polynomial time, otherwise $\CSP(\mathbb{A})$ is
  NP-complete.
\end{algdichotomy}
\noindent Indeed, perhaps the above conjecture should be called the
\emph{algebraic tractability conjecture} since it is known that if the algebra of 
polymorphisms of a
core $\mathbb A$ is not Taylor, then
$\CSP(\mathbb A)$ is NP-complete~\cite{b-j-k}.

Larose and Tesson \cite{lartes} conjectured a similar characterization of
finite relational structures with the corresponding CSP solvable in L and in NL. 
In the same paper they also proved the hardness part of boths claims.
Their conjecture is widely discussed in the following slightly stronger form (equivalent modulo reasonable complexity-theoretic assumptions; see the discussion in \cite{JKN13}).

\begin{finerCSPcomplexityconjectures} Let $\mathbb{A}$ be a finite relational 
structure that is a core. Then the following hold.
\begin{enumerate}[label=(\roman*)]
\item 
$\CSP(\mathbb{A})$ is solvable in nondeterministic logspace, if and only if 
the algebra of polymorphisms of $\mathbb{A}$ is congruence join-semidistributive.
\item 
$\CSP(\mathbb{A})$ is solvable in logspace, if and only if the algebra of 
polymorphisms of $\mathbb{A}$ is congruence join-semidistributive
and congruence $n$-permutable for some $n$.
\end{enumerate}
\end{finerCSPcomplexityconjectures}

\noindent 
Feder and Vardi~\cite{fedvar} proved that every fixed template CSP is
polynomial time equivalent to a digraph CSP. Thus the CSP dichotomy conjecture is 
equivalent to its restriction to digraphs. In this paper we investigate a construction similar to theirs. 
The main results of this paper are summarized in the following theorem.

\begin{thm}\label{thm:mainresults}
  For every finite relational structure $\mathbb{A}$ there exists a finite
  digraph~$\mathcal D(\mathbb A)$ such that the following holds:
  \begin{enumerate}[label=(\roman*)]
  \item $\CSP(\mathbb{A})$ and $\CSP(\mathcal D(\mathbb A))$ are
    logspace equivalent.
  \item $\mathbb{A}$ is a core if and only if $\mathcal D(\mathbb A)$
    is a core.
  \item If $\Sigma$ is a
  linear idempotent set of identities such that the algebra of
  polymorphisms of the oriented path
  \mbox{$\bullet\boldsymbol\rightarrow\bullet\boldsymbol
    \leftarrow\bullet\boldsymbol\rightarrow\bullet$} satisfies
  $\Sigma$ and each identity in $\Sigma$ is either balanced or
  contains at most two variables, then
$$
\mathbb A\models\Sigma\text{ if and only if }\mathcal D(\mathbb A)\models\Sigma.
$$
  
 \end{enumerate}
\end{thm} 

\begin{proof}
  Item (i) is Theorem~\ref{thm:logspace},  (ii) is Corollary~\ref{cor:core} and (iii)
  is Theorem \ref{thm:preserved_conditions}.
\end{proof}

\noindent The construction of $\mathcal D(\mathbb A)$ is described in 
Section \ref{sec:reduction}, for a bound on the size of $\mathcal D(\mathbb A)$ 
see Proposition~\ref{rem:number}. The condition on $\Sigma$ in item (iii) is not very 
restrictive: it includes 
almost all of the commonly encountered properties relevant to the CSP. 
A number of these are listed in Corollary \ref{cor:preserved_conditions}.  
Note that the list includes
the properties of being Taylor, congruence join-semidistributive and congruence 
$n$-permutable (for $n\geq 3$); hence we have the following corollary.

\begin{cor}
The Algebraic CSP dichotomy conjecture and the Finer CSP complexity conjectures are 
also equivalent to their restrictions to digraphs. 
\end{cor}

\section{Background and definitions}\label{sec:definitions}
\noindent We approach fixed template constraint satisfaction problems
from the ``homomorphism problem'' point of view. For background on the
homomorphism approach to CSPs, see \cite{fedvar}, and for background
on the algebraic approach to CSPs, see \cite{b-j-k}.

A \emph{relational signature} $\mathcal{R}$ is a (in our case finite)
set of \emph{relation symbols} $R_i$, each with an associated arity
$k_i$. A (finite) \emph{relational structure} $\mathbb{A}$ \emph{over
  relational signature} $\mathcal{R}$ (called an
\emph{$\mathcal{R}$-structure}) is a finite set $A$ (the
\emph{domain}) together with a relation $R_i\subseteq A^{k_i}$, for
each relation symbol $R_i$ of arity $k_i$ in $\mathcal{R}$. A
\emph{CSP template} is a fixed finite $\mathcal{R}$-structure, for
some signature $\mathcal{R}$.

For simplicity we do not distinguish the relation with its associated
relation symbol. However, to avoid ambiguity, we sometimes write
$R^\mathbb A$ to indicate that $R$ is interpreted in $\mathbb A$. We will
often refer to the domain of a relational structure $\mathbb A$ simply
by $A$. When referring to a fixed relational structure, we may simply
specify it as $\mathbb A = (A; R_1,R_2,\dots,R_n)$. For technical
reasons we require that signatures are nonempty and that all the relations of a relational structure are nonempty.

\subsection{Notation}\label{sec:notation}
\noindent For a positive integer $n$ we denote the set
$\{1,2,\dots,n\}$ by $[n]$. We write tuples using boldface notation,
e.g. $\mathbf a=(a_1,a_2,\dots,a_k)\in A^k$ and when ranging over
tuples we use superscript notation,
e.g. $(\mathbf{r}^1,\mathbf{r}^2,\dots,\mathbf{r}^l)\in R^l\subseteq
(A^k)^l$, where $\mathbf{r}^i=(r^i_1,r^i_2,\dots,r^i_k)$, for
$i=1,\dots,l$.

Let $R_i\subseteq A^{k_i}$ be relations of arity $k_i$, for
$i=1,\dots,n$. Let $k=\sum_{i=1}^nk_i$ and $l_i=\sum_{j<i}k_j$. We
write $R_1\times\dots\times R_n$ to mean the $k$-ary relation
\[
\{ (a_1,\dots, a_k)\in A^k \mid (a_{l_i+1},\dots,a_{l_i+k_i})\in R_i
\text{ for } i=1,\dots,n\}.
\]

An \emph{$n$-ary operation} on a set $A$ is simply a mapping
$f:A^n\rightarrow A$; the number $n$ is the \emph{arity} of $f$.  Let
$f$ be an $n$-ary operation on $A$ and let $k>0$. We write $f^{(k)}$
to denote the $n$-ary operation obtained by applying $f$ coordinatewise on
$A^k$. That is, we define the $n$-ary operation $f^{(k)}$ on $A^k$ by
\[
f^{(k)}(\mathbf a^1,\dots,\mathbf
a^n)=(f(a^1_1,\dots,a^n_1),\dots,f(a^1_k,\dots,a^n_k)),
\]
for $\mathbf a^1,\dots, \mathbf a^n\in A^k$.

We will be particularly interested in so-called idempotent operations.
An $n$-ary operation $f$ is said to be \emph{idempotent} if it
satisfies the equation $f(x,x,\dots,x)=x$.

\subsection{Homomorphisms, cores and polymorphisms}\label{sec:homs}
\noindent We begin with the notion of a relational structure
homomorphism.
\begin{defi}\label{def:hom}
  Let $\mathbb A$ and $\mathbb B$ be relational structures in the same
  signature~$\mathcal{R}$. A \emph{homomorphism} from $\mathbb A$ to
  $\mathbb B$ is a mapping $\varphi$ from $A$ to $B$ such that for
  each $k$-ary relation symbol $R$ in $\mathcal{R}$ and each $k$-tuple
  $\mathbf{a}\in A^k$, if $\mathbf{a}\in R^\mathbb A$, then
  $\varphi^{(k)}(\mathbf{a})\in R^\mathbb B$.
\end{defi}

We write $\varphi:\mathbb A\to\mathbb B$ to mean that $\varphi$ is a
homomorphism from $\mathbb A$ to $\mathbb B$, and $\mathbb A\to\mathbb
B$ to mean that there exists a homomorphism from $\mathbb A$ to
$\mathbb B$. 

An \emph{isomorphism} is a bijective homomorphism $\varphi$ such that
$\varphi^{-1}$ is also a homomorphism. A homomorphism $\mathbb A\to\mathbb
A$ is called an \emph{endomorphism}. An isomorphism from $\mathbb A$
to $\mathbb A$ is an \emph{automorphism}. It is an easy fact that if
$\mathbb A$ is finite, then every surjective endomorphism is an
automorphism.

A finite relational structure $\mathbb A'$ is a \emph{core} if every
endomorphism $\mathbb A'\to\mathbb A'$ is surjective (and therefore an
automorphism). For every $\mathbb A$ there exists a relational
structure $\mathbb A'$ such that $\mathbb A\to\mathbb A'$ and $\mathbb
A'\to\mathbb A$ and $\mathbb A'$ is of minimum size with respect to these
properties; that structure $\mathbb A'$ is called the \emph{core of
  $\mathbb A$}. The core of $\mathbb A$ is unique (up to isomorphism)
and $\CSP(\mathbb A)$ and $\CSP(\mathbb A')$ are the same decision
problems. Equivalently, the core of $\mathbb A$ can be defined as an induced substructure of minimum size that $\mathbb A$ retracts onto. (See~\cite{helnes} for details on cores for graphs, cores for
relational structures are a natural generalization.)

The notion of \emph{polymorphism} is central in the 
so-called 
algebraic approach to the $\CSP$. Polymorphisms are a natural
generalization of endomorphisms to higher arity operations.

\begin{defi}
  Given an $\mathcal{R}$-structure $\mathbb{A}$, an $n$-ary
  \emph{polymorphism} of $\mathbb{A}$ is an $n$-ary operation $f$ on
  $A$ such that $f$ preserves the relations of $\mathbb A$. That is,
  if $\mathbf{a}^1,\dots,\mathbf{a}^n\in R$, for some $k$-ary relation
  $R$ in $\mathcal{R}$, then $f^{(k)}(\mathbf a^1,\dots,\mathbf
  a^n)\in R$.  
\end{defi}
Thus, an endomorphism is a unary polymorphism. Polymorphisms satisfying certain identities has been used extensively 
in the algebraic study of CSPs.

\subsection{Algebra} \label{subsection:preliminaries}
\noindent Given a finite relational structure $\mathbb A$, let
$\Pol\mathbb A$ denote the set of all polymorphisms of $\mathbb
A$. The \emph{algebra of polymorphisms} of $\mathbb A$ is simply the
algebra with the same universe whose operations are all polymorphisms
of $\mathbb A$. A subset $B\subseteq A$ is a \emph{subuniverse} of
$\mathbb A$, denoted by $B\leq\mathbb A$, if it is a subuniverse of
the algebra of polymorphisms of $\mathbb A$, i.e., it is closed under
all $f\in\Pol\mathbb A$.

An \emph{(operational) signature} is a (possibly infinite) set of
operation symbols with arities assigned to them. By an \emph{identity}
we mean an expression $u\approx v$ where $u$ and $v$ are terms in some
signature. An identity $u\approx v$ is \emph{linear} if both $u$ and
$v$ involve at most one occurrence of an operation symbol
(e.g. $f(x,y)\approx g(x)$, or $h(x,y,x)\approx x$); and
\emph{balanced} if the sets of variables occuring in $u$ and in $v$
are the same (e.g. $f(x,x,y)\approx g(y,x,x)$).

A set of identities $\Sigma$ is \emph{linear} if it contains only
linear identities; \emph{balanced} if all the identities in $\Sigma$
are balanced; and \emph{idempotent} if for each operation symbol $f$
appearing in an identity of $\Sigma$, the identity
$f(x,x,\dots,x)\approx x$ is in $\Sigma$. \footnote{We can relax this
  condition and require the identity $f(x,x,\dots,x)\approx x$ only to
  be a \emph{syntactical consequence} of identities in $\Sigma$.} For example, the
identities $p(y,x,x)\approx y,\ p(x,x,y)\approx y,\ p(x,x,x)\approx x$
(defining the so-called \emph{Maltsev} operation) form a linear idempotent
set of identities which is not balanced.

A~\emph{strong Maltsev condition}, commonly encountered in universal
algebra, can be defined in this context as a finite set of
identities. A \emph{Maltsev condition} is an increasing chain of
strong Maltsev conditions, ordered by syntactical consequence. In all
results from this paper, ``set of identities'' can be replaced with
``Maltsev condition''.

Let $\Sigma$ be a set of identities in a signature with operation
symbols $\mathcal{F}=\{ f_\lambda\mid\lambda\in\Lambda\}$. We say that
a relational structure $\mathbb{A}$ \emph{satisfies} $\Sigma$ (and
write \emph{$\mathbb A\models\Sigma$}), if for every
$\lambda\in\Lambda$ there is a polymorphism
$f^\mathbb{A}_\lambda\in\Pol\mathbb A$ such that the identities in
$\Sigma$ hold universally in $\mathbb A$ when for each
$\lambda\in\Lambda$ the symbol $f_\lambda$ is interpreted as
$f^\mathbb A_\lambda$.

For example, a \emph{weak near-unanimity} (\emph{WNU}) is an $n$-ary ($n\geq 2$) 
idempotent operation $\omega$ satisfying the identities
\[
\omega(x,\dots,x,y)=\omega(x,\dots,x,y,x)=\dots=\omega(y,x,\dots,x).
\]
Thus, having an $n$-ary weak near-unanimity is definable by a 
linear balanced idempotent set of identities. 
Existence of WNU polymorphisms influences $\mathrm{CSP}(\mathbb A)$ to a great 
extent. The following characterization was discovered in \cite{maroti-mckenzie}: a finite algebra (or relational structure) is
\begin{itemize}
 \item \emph{Taylor}, if it has a weak near-unanimity operation of some arity, and 
 \item \emph{congruence meet-semidistributive} if it has WNU operations of all 
 but finitely many arities.
\end{itemize}
The Algebraic CSP dichotomy conjecture asserts that being Taylor is what 
distinguishes tractable (core) relational structures from the NP-complete ones, and
a similar split is known for congruence meet-semidistributivity and solvability by local consistency checking 
(the so-called \emph{bounded width}):

\begin{boundedwidththeorem}{\rm \cite{BW_journal_version}} \label{thm:bounded_width}
Let $\mathbb A$ be a finite relational structure that is a
core. Then $\CSP(\mathbb A)$ is solvable by local consistency checking, if and only if the algebra of polymorphisms
of $\mathbb A$ is congruence meet-semidistributive.
\end{boundedwidththeorem}

The properties of \emph{congruence join-semidistributivity} and 
\emph{congruence $n$-permutability} which appear in the finer CSP complexity 
conjectures are also definable by linear idempotent sets of identities, 
albeit more complicated ones; we refer the reader to \cite{hobbymckenzie}. 
We will introduce more Maltsev conditions and their connection to the CSP in Section \ref{sec:Maltsev}.

\subsection{Primitive positive definability}\label{sec:rel struct}
\noindent A first order formula is called \emph{primitive positive} if
it is an existential conjunction of atomic formul\ae.  Since we only
refer to relational signatures, a primitive positive formula is simply
an existential conjunct of formul\ae\ of the form $x = y$ or
$(x_1,x_2,\dots,x_k) \in R$, where $R$ is a relation symbol of arity
$k$.

For example, if we have a binary relation symbol $E$ in our signature,
then the formula
\[
\psi(x,y) = (\exists z)((x,z)\in E\ \wedge\ (z,y)\in E)
\]
pp-defines a binary relation in which elements $a,b$ are related if
there is a directed path of length $2$ from $a$ to $b$ in $E$.

\begin{defi}
  A relational structure $\mathbb{B}$ is \emph{primitive positive
    definable} in $\mathbb{A}$ \textup(or~$\mathbb{A}$
  \emph{pp-defines} $\mathbb{B}$\textup) if
\begin{enumerate}[label=(\roman*)]
\item the set $B$ is a subset of $A$ and is definable by a primitive
  positive formula interpreted in $\mathbb{A}$, and
\item each relation $R$ in the signature of $\mathbb{B}$ is definable
  on the set $B$ by a primitive positive formula interpreted in
  $\mathbb{A}$.
\end{enumerate}
\end{defi}

\noindent The following result relates the above definition to the complexity of
CSPs.  The connection is originally due to Jeavons, Cohen and Gyssens \cite{JCG97}, though the logspace form stated and used here can be found in Larose and Tesson \cite[Theorem~2.1]{lartes}.
\begin{lem}\label{lem:logspace_reduction}
  Let $\mathbb{A}$ be a finite relational structure that pp-defines
  $\mathbb{B}$. Then, $\mathrm{CSP}(\mathbb{B})$ is logspace reducible to $\CSP(\mathbb{A})$.
\end{lem}

It so happens that, if $\mathbb{A}$ pp-defines $\mathbb{B}$, then
$\mathbb{B}$ inherits the polymorphisms of
$\mathbb{A}$. See~\cite{b-j-k} for a detailed explanation. 
\begin{lem}{\rm \cite{b-j-k}}\label{lem:ppdef_polymorphisms}
  Let $\mathbb{A}$ be a finite relational structure that pp-defines
  $\mathbb{B}$. If $\varphi$ is a polymorphism of $\mathbb A$, then
  its restriction to $B$ is a polymorphism of $\mathbb B$.
\end{lem}

\noindent In particular, as an easy consequence of this lemma,
if $\mathbb A$ pp-defines $\mathbb B$ and $\mathbb A$ satisfies
a set of identities $\Sigma$, then $\mathbb B$ also satisfies $\Sigma$.

In the case that $\mathbb{A}$ pp-defines $\mathbb{B}$ and $\mathbb{B}$
pp-defines $\mathbb{A}$, we say that $\mathbb{A}$ and $\mathbb{B}$ are
\emph{pp-equivalent}. In this case, $\CSP(\mathbb{A})$ and
$\CSP(\mathbb{B})$ are essentially the same problems (they are
logspace equivalent) and $\mathbb{A}$ and $\mathbb{B}$ have the same
polymorphisms.

\begin{exa}\label{example:1}
  Let $\mathbb{A}=(A;R_1,\dots,R_n)$, where each $R_i$ is $k_i$-ary,
  and define $R=R_1\times \dots \times R_n$. Then the structure
  $\mathbb{A}'=(A;R)$ is pp-equivalent to $\mathbb{A}$.  

  Indeed, let $k=\sum_{i=1}^nk_i$ be the arity of $R$ and
  $l_i=\sum_{j<i}k_j$ for $i=1,\dots,n$. The relation $R$ is
  pp-definable from $R_1,\dots,R_n$ using the formula
  \[
  \Psi(x_1,\dots,x_k)= \bigwedge_{i=1}^n
  (x_{l_i+1},\dots,x_{l_i+k_i})\in R_i.
  \]
  The relation $R_1$ can be defined from $R$ by the primitive positive
  formula
  \[
  \Psi(x_1,\dots,x_{k_1})=(\exists y_{k_1+1},\dots,\exists y_k)
  ((x_1,\dots,x_{k_1},y_{k_1+1},\dots,y_k)\in R)
  \]
  and the remaining $R_i$'s can be defined similarly.
\end{exa}
Example \ref{example:1} shows that when proving Theorem \ref{thm:mainresults} we can
restrict ourselves to relational structures with a single relation.

\subsection{Digraphs}
\noindent A \emph{directed graph}, or \emph{digraph}, is a relational
structure $\mathbb{G}$ with a single binary relation symbol $E$ as
its signature. We typically call the members of $G$ and $E^\mathbb{G}$
\emph{vertices} and \emph{edges}, respectively. We usually write $a\to
b$ to mean $(a,b)\in E^\mathbb{G}$, if there is no ambiguity.

A special case of relational structure homomorphism (see
Definition~\ref{def:hom}), is that of digraph homomorphism. That is,
given digraphs $\mathbb G$ and $\mathbb H$, a function $\varphi:G\to
H$ is a homomorphism if $(\varphi(a),\varphi(b))\in E^\mathbb{H}$
whenever $(a,b)\in E^\mathbb{G}$.

\begin{defi}
  For $i=1,\dots, n$, let $\mathbb{G}_i=(G_i,E_i)$ be digraphs. The
  \emph{direct product of} $\mathbb{G}_1,\dots,\mathbb{G}_n$, denoted
  by $\prod_{i=1}^n \mathbb{G}_i$, is the digraph with vertices
  $\prod_{i=1}^nG_i$ {\rm(}the cartesian product of the sets $G_i${\rm
    )} and edge relation
  \[
  \{(\mathbf{a},\mathbf{b})\in (\prod_{i=1}^nG_i)^2 \mid
  (a_i,b_i)\in E_i \text{ for } i=1\dots,n \}.
  \]
  If $\mathbb{G}_1=\dots = \mathbb{G}_n =\mathbb{G}$ then we write
  $\mathbb{G}^n$ to mean $\prod_{i=1}^n \mathbb{G}_i$.
\end{defi}

With the above definition in mind, an $n$-ary polymorphism on a
digraph $\mathbb{G}$ is simply a digraph homomorphism from
$\mathbb{G}^n$ to $\mathbb{G}$. 

\begin{defi}
  A digraph $\mathbb P$ is an \emph{oriented path} if it consists of a
  sequence of vertices $v_0,v_1,\dots,v_k$ such that for each $i=1,\dots,k$ precisely one of $(v_{i-1},v_i),(v_i,v_{i-1})$ is an edge, and there are no other edges. We require oriented paths to have a direction; we denote the
  \emph{initial} vertex $v_0$ and the \emph{terminal} vertex $v_k$ by
  $\iota\mathbb P$ and $\tau\mathbb P$, respectively.
\end{defi}

Given a digraph $\mathbb G$ and an oriented path $\mathbb P$, we write
$a\stackrel{\mathbb{P}}{\longrightarrow} b$ to mean that we can walk
in $\mathbb G$ from $a$ following $\mathbb{P}$ to $b$, i.e., there
exists a homomorphism $\varphi:\mathbb P\to\mathbb G$ such that
$\varphi(\iota\mathbb P)=a$ and $\varphi(\tau\mathbb P)=b$. Note that
for every $\mathbb P$ there exists a primitive positive formula
$\psi(x,y)$ such that $a\stackrel{\mathbb{P}}{\longrightarrow} b$ if
and only if $\psi(a,b)$ is true in $\mathbb{G}$. If there exists an
oriented path $\mathbb P$ such that
$a\stackrel{\mathbb{P}}{\longrightarrow} b$, we say that $a$ and $b$
are \emph{connected}. If vertices $a$ and $b$ are connected, then the
\emph{distance} from $a$ to $b$ is the number of edges in the shortest
oriented path connecting them. Connectedness forms an equivalence
relation on $G$; its classes are called the \emph{connected
  components} of $\mathbb G$. We say that a digraph is \emph{connected} if it
consists of a single connected component.\footnote{The notions of connectedness and distance are the same as in the undirected graph obtained by forgetting orientation of edges of $\mathbb G$.}

A connected digraph is \emph{balanced} if it admits a \emph{level
  function} $\lvl:G\to\mathbb N\cup\{0\}$, where
$\lvl(b)=\lvl(a)+1$ whenever $(a,b)$ is an edge, and
the minimum level is~$0$. The maximum level is called the
\emph{height} of the digraph. Oriented paths are natural examples of
balanced digraphs.

By a \emph{zigzag} we mean the oriented path
\mbox{$\bullet\boldsymbol\rightarrow\bullet\boldsymbol
  \leftarrow\bullet\boldsymbol\rightarrow\bullet$} and a \emph{single
  edge} is the path \mbox{$\bullet\boldsymbol\rightarrow\bullet$}. For
oriented paths $\mathbb P$ and $\mathbb P'$, the \emph{concatenation
  of $\mathbb P$ and $\mathbb P'$}, denoted by $\mathbb
P\dotplus\mathbb P'$, is the oriented path obtained by identifying
$\tau\mathbb P$ with $\iota\mathbb P'$.

Our digraph reduction as described in Section~\ref{sec:reduction}
relies on oriented paths obtained by concatenation of zigzags and
single edges. For example, the path in Figure~\ref{fig:path} is a
concatenation of a single edge followed by two zigzags and two more
single edges (for clarity, we organize its vertices into levels).
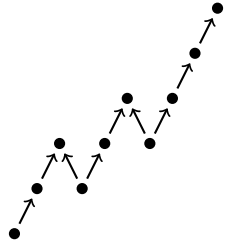
\begin{figure}[h]
\[
\begin{tikzpicture}[thick, on grid, node distance=.6cm and
  .3cm,dot/.style={circle,draw,outer sep=2.5pt,inner sep=0pt,minimum
    size=1.2mm,fill}]
\node [dot] (a) at (0,0) {};
\node [dot,above right=of a] (b)  {};
\node [dot,above right=of b] (c)  {};
\node [dot,below right=of c] (d)  {};
\node [dot,above right=of d] (e)  {};
\node [dot,above right=of e] (f)  {};
\node [dot,below right=of f] (g)  {};
\node [dot,above right=of g] (h)  {};
\node [dot,above right=of h] (i)  {};
\node [dot,above right=of i] (j)  {};
\draw[->] (a) -- (b);
\draw[->] (b) -- (c);
\draw[<-] (c) -- (d);
\draw[->] (d) -- (e);
\draw[->] (e) -- (f);
\draw[<-] (f) -- (g);
\draw[->] (g) -- (h);
\draw[->] (h) -- (i);
\draw[->] (i) -- (j);
\end{tikzpicture}
\]
\caption{A minimal oriented path}\label{fig:path}
\end{figure}

\section{The reduction to digraphs}\label{sec:reduction}
\noindent In this section we take an arbitrary finite relational
structure $\mathbb{A}$ and construct a balanced digraph $\mathcal
D(\mathbb A)$ such that $\CSP(\mathbb{A})$ and $\CSP(\mathcal
D(\mathbb A))$ are logspace equivalent.

Let $\mathbb A=(A;R_1,\dots,R_n)$ be a finite relational structure,
where $R_i$ is of arity~$k_i$, for $i=1,\dots,n$. Let $k=\sum_{i=1}^n
k_i$ and let $R$ be the $k$-ary relation $R_1\times\dots\times
R_n$. For $\mathcal I\subseteq [k]$ define $\mathbb Q_{\mathcal I,l}$ to be a
single
edge if $l\in\mathcal I$, and a zigzag if $l\in[k]\setminus\mathcal I$.

We define the oriented path
$\mathbb Q_\mathcal I$ (of height $k+2$) by
\[
\mathbb Q_\mathcal I\,=\, \mbox{$\bullet\boldsymbol\rightarrow\bullet$}
\dotplus\,\mathbb Q_{\mathcal I,1}\,\dotplus\,\mathbb
Q_{\mathcal I,2}\,\dotplus\,\dots\,\dotplus\,\mathbb Q_{\mathcal
I,k}\,\dotplus\,
\mbox{$\bullet\boldsymbol\rightarrow\bullet$}
\]
Instead of $\mathbb Q_\emptyset,\mathbb Q_{\emptyset,l}$ we write just
$\mathbb Q,\mathbb Q_l$, respectively. For example, the oriented path
in Figure~\ref{fig:path} is $\mathbb Q_{\mathcal I}$ where $k=3$ and
$\mathcal I=\{3\}$. We will need the following observation.

\begin{obs}
  Let $\mathcal I,\mathcal J\subseteq [k]$. A homomorphism $\varphi:\mathbb
  Q_\mathcal I\to\mathbb Q_\mathcal J$ exists, if and only if $\mathcal
I\subseteq\mathcal J$. In
  particular $\mathbb Q\to \mathbb Q_\mathcal I$ for all $\mathcal I\subseteq
[k]$.
  Moreover, if $\varphi$ exists, it is unique and surjective.
\end{obs}

We are now ready to define the digraph $\mathcal D(\mathbb A)$.
\begin{defi}
  For every $e=(a,\mathbf r)\in A\times R$ we define $\mathbb{P}_e$ to
  be the path $\mathbb Q_{\{i\,\mid\,a=r_i\}}$.  The digraph $\mathcal
  D(\mathbb A)$ is obtained from the digraph $(A\cup R;A\times R)$ by
  replacing every $e=(a,\mathbf r)\in A\times R$ by the oriented path
  $\mathbb{P}_e$ \up(identifying $\iota\mathbb P_e$ with $a$ and
  $\tau\mathbb P_e$ with $\mathbf r$\up).
\end{defi}
(We often write $\mathbb P_{e,l}$ to mean $\mathbb Q_{\mathcal{I},l}$
where $\mathbb P_e=\mathbb Q_\mathcal{I}$.)

\begin{exa}
  Consider the relational structure $\mathbb A=(\{ 0,1\}; R)$ where
  $R=\{ (0,1),(1,0)\}$, i.e., $\mathbb A$ is the directed $2$-cycle. Figure~\ref{fig:cycle} is a visual representation of
  $\mathcal D(\mathbb A)$.
\begin{figure}[h]
\[
\begin{tikzpicture}[thick, on grid, node distance=.8cm and
  .4cm,
  dot/.style={circle,draw,outer sep=2.5pt,inner sep=0pt,minimum
size=1.2mm,fill},
  emptydot/.style={circle,draw,outer sep=2.5pt,inner sep=0pt,minimum
size=1.2mm,stroke},
  dot2/.style={circle,draw,outer sep=2.5pt,inner sep=0pt,minimum
size=1.2mm,fill,color=gray},
  emptydot2/.style={circle,draw,outer sep=2.5pt,inner sep=0pt,minimum
size=1.2mm,stroke,color=gray}
]
\node [dot] (a) at (-.6,0) {};
\node [below=3mm of a]{$0$};
\node [emptydot,above right=of a] (b)  {};
\node [emptydot,above right=of b] (e)  {};
\node [emptydot,above right=of e] (f)  {};
\node [emptydot,below right=of f] (g)  {};
\node [emptydot,above right=of g] (h)  {};
\node [dot,above right=of h] (i)  {};
\node [above=3mm of i]{(0,1)};

\draw[->] (a) -- (b);
\draw[->] (b) -- (e);
\draw[->] (e) -- (f);
\draw[<-] (f) -- (g);
\draw[->] (g) -- (h);
\draw[->] (h) -- (i);

\node [dot] (a2) at (.6,0) {};
\node [below=3mm of a2]{$1$};
\node [emptydot2,above left=of a2] (b2)  {};
\node [emptydot2,above left=of b2] (e2)  {};
\node [emptydot2,above left=of e2] (f2)  {};
\node [emptydot2,below left=of f2] (g2)  {};
\node [emptydot2,above left=of g2] (h2)  {};
\node [dot,above left=of h2] (i2)  {};
\node [above=3mm of i2]{(1,0)};

\draw[->,color=gray] (a2) -- (b2);
\draw[->,color=gray] (b2) -- (e2);
\draw[->,color=gray] (e2) -- (f2);
\draw[<-,color=gray] (f2) -- (g2);
\draw[->,color=gray] (g2) -- (h2);
\draw[->,color=gray] (h2) -- (i2);

\node [emptydot2,above right=of a2] (b2)  {};
\node [emptydot2,above right=of b2] (e2)  {};
\node [emptydot2,below right=of e2] (f2)  {};
\node [emptydot2,above right=of f2] (g2)  {};
\node [emptydot2,above =of g2] (h2)  {};

\draw[->,color=gray] (a2) -- (b2);
\draw[->,color=gray] (b2) -- (e2);
\draw[->,color=gray] (f2) -- (e2);
\draw[<-,color=gray] (g2) -- (f2);
\draw[->,color=gray] (g2) -- (h2);
\draw[->,color=gray] (h2) -- (i);

\node [emptydot,above left=of a] (b2)  {};
\node [emptydot,above left=of b2] (e2)  {};
\node [emptydot,below left=of e2] (f2)  {};
\node [emptydot,above left=of f2] (g2)  {};
\node [emptydot,above =of g2] (h2)  {};

 \draw[->] (a) -- (b2);
\draw[->] (b2) -- (e2);
\draw[->] (f2) -- (e2);
\draw[<-] (g2) -- (f2);
\draw[->] (g2) -- (h2);
\draw[->] (h2) -- (i2);

\end{tikzpicture}
\]
\caption{$\mathcal D(\mathbb A)$ where $\mathbb A$ is the directed
  $2$-cycle}\label{fig:cycle}
\end{figure}
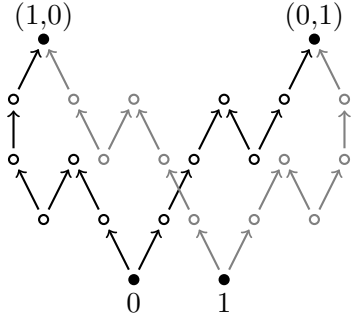
\end{exa}

\begin{prop}\label{rem:number}
  The number of vertices in $\mathcal D(\mathbb A)$ is
  $(3k+1)|R||A|+(1-2k)|R|+|A|$, the number of edges is
  $(3k+2)|R||A|-2k|R|$ and the height is $k+2$. The construction of
  $\mathcal D(\mathbb A)$ can be performed in logspace \up(under any
  reasonable encoding\up).
\end{prop}
\begin{proof} The vertices of $\mathcal D(\mathbb A)$ consist of the
  elements of $A\cup R$, along with vertices from the connecting
  paths.  The number of vertices lying strictly within the connecting
  paths would be $(3k+1)|R||A|$ if every $\mathbb P_e$ was $\mathbb
  Q$. We need to deduct~$2$ vertices whenever there is a single edge
  instead of a zigzag and there are $\sum_{(a,\mathbf r)\in A\times
    R}|\{i\,\mid\,a=r_i\}|=k|R|$ such instances.  The number of edges is
  counted very similarly.
\end{proof}

\begin{rem}
  As $\mathcal{D}(\mathbb A)$ is always a digraph, the construction
  $\mathcal D(\mathcal D(\mathbb A))$ is  digraph of height 4 and also relates to $\mathbb A$ on all of the conditions described in Theorem \ref{thm:mainresults}.   When applied to digraphs, the $\mathcal{D}$
  construction is identical to that given by Feder and
  Vardi~\cite[Theorem 13]{fedvar}.
\end{rem}
The next lemma, together with Lemma~\ref{lem:logspace_reduction},
shows that $\CSP(\mathbb{A})$ reduces to $\CSP(\mathcal D(\mathbb A))$
in logspace.
\begin{lem}\label{lem:forwardreduction}
$\mathbb{A}$ is pp-definable from $\mathcal D(\mathbb A)$.
\end{lem}
\begin{proof}
  Example~\ref{example:1} demonstrates that $\mathbb A$ is
  pp-equivalent to $(A;R)$. We now show that $\mathcal D(\mathbb A)$
  pp-defines $(A;R)$, from which it follows that $\mathcal D(\mathbb
  A)$ pp-defines $\mathbb{A}$.

  Note that $\mathbb Q\to\mathbb P_e$ for all $e\in A\times R$, and
  $\mathbb Q_{\{i\}}\to\mathbb P_{(a,\mathbf r)}$ if and only if
  $a=r_i$.  The set $A$ is pp-definable in $\mathcal D(\mathbb A)$ by
  $A=\{x\mid (\exists y)(x\stackrel{\mathbb Q}{\longrightarrow} y)\}$
  and the relation~$R$ can be defined as the set
  $\{(x_1,\dots,x_k)\mid (\exists y)(x_i\stackrel{\mathbb
    Q_{\{i\}}}{\longrightarrow}y \text{ for all }i\in[k])\}$, which is
  also a primitive positive definition.
\end{proof}

It is not, in general, possible to pp-define $\mathcal D(\mathbb A)$
from $\mathbb{A}$.\footnote[1]{Using the definition of pp-definability
  as described in this paper, this is true for cardinality
  reasons. However, a result of Kazda \cite{kaz} can be used to show
  that the statement remains true even for more general definitions of
  pp-definability.}  Nonetheless the following lemma is true.

\begin{lem}\label{lem:reversereduction} 
  $\CSP(\mathcal D(\mathbb A))$ reduces in logspace to
  $\CSP(\mathbb{A})$.
\end{lem}
The proof of Lemma \ref{lem:reversereduction} is rather technical,
though broadly follows the polynomial process described in the proof
of \cite[Theorem 13]{fedvar} (as mentioned, our construction coincides
with theirs in the case of digraphs).  Details of the argument are
provided in Section~\ref{sec:reversereduction}.  

\section{Preserving cores}
\noindent In what follows, let $\mathbb A$ be a fixed finite
relational structure. Without loss of generality we may assume that
$\mathbb A=(A;R)$, where $R$ is a $k$-ary relation (see
Example~\ref{example:1}).

\begin{lem}\label{lem:end} 
  The endomorphisms of $\mathbb A$ and $\mathcal{D}(\mathbb A)$ are in
  one-to-one correspondence.
\end{lem}
\begin{proof} 
  We first show that every endomorphism $\varphi$ of $\mathbb A$ can
  be extended to an endomorphism $\overline{\varphi}$ of $\mathcal
  D(\mathbb A)$.  Let $\overline{\varphi}(a)=\varphi(a)$ for $a\in A$,
  and let $\overline{\varphi}(\mathbf r)=\varphi^{(k)}(\mathbf r)$ for
  $\mathbf r\in R$. Note that $\varphi^{(k)}(\mathbf r)\in R$ since
  $\varphi$ is an endomorphism of $\mathbb A$.

  Let $c\in\mathcal D(\mathbb A)\setminus(A\cup R)$ and let
  $e=(a,\mathbf r)$ be such that $c\in\mathbb P_e$. Define
  $e'=(\varphi(a),\varphi^{(k)}(\mathbf r))$. If $\mathbb P_{e,l}$ is
  a single edge for some $l\in[k]$, then $r_l=a$ and
  $\varphi(r_l)=\varphi(a)$, and therefore $\mathbb P_{e',l}$ is a
  single edge. Thus there exists a (unique) homomorphism $\mathbb
  P_e\to\mathbb P_{e'}$. Define $\overline{\varphi}(c)$ to be the
  image of $c$ under this homomorphism, completing the definition of
  $\overline{\varphi}$.

  We now show that every endomorphism $\Phi$ of $\mathcal{D}(\mathbb
  A)$ is of the form $\overline{\varphi}$, for some endomorphism
  $\varphi$ of $\mathbb A$.  Let $\Phi$ be an endomorphism of
  $\mathcal D(\mathbb A)$. Let $\varphi$ be the restriction of $\Phi$
  to $A$. By Lemma \ref{lem:ppdef_polymorphisms} and Lemma
  \ref{lem:forwardreduction}, $\varphi$ is an endomorphism of $\mathbb
  A$.  For every $e=(a,\mathbf r)$, the endomorphism $\Phi$ maps
  $\mathbb P_e$ onto $\mathbb P_{(\varphi(a),\Phi(\mathbf r))}$. If we
  set $a=r_l$, then $\mathbb P_{e,l}$ is a single edge. In this case
  it follows that $\mathbb P_{(\varphi(a),\Phi(\mathbf r)),l}$ is also
  a single edge. Thus, by the construction of $\mathcal D(\mathbb A)$
  the $l^\text{th}$ coordinate of $\Phi(\mathbf r)$ is
  $\varphi(a)=\varphi(r_l)$. This proves that the restriction of
  $\Phi$ to $R$ is $\varphi^{(k)}$ and therefore
  $\Phi=\overline{\varphi}$.
\end{proof}

The following corollary is Theorem~\ref{thm:mainresults} (ii).

\begin{cor}\label{cor:core} 
  $\mathbb{A}$ is a core if and only if $\mathcal D(\mathbb A)$ is a
  core.
\end{cor}

\begin{proof}
  To prove the corollary we need to show that an endomorphism
  $\varphi$ of $\mathbb A$ is surjective if and only if
  $\overline{\varphi}$ (from Lemma~\ref{lem:end}) is
  surjective. Clearly, if $\overline{\varphi}$ is surjective then so
  is $\varphi$.

  Assume $\varphi$ is surjective (and therefore an automorphism of
  $\mathbb A$). It follows that $\varphi^{(k)}$ is surjective on $R$ and
therefore
  $\overline{\varphi}$ is a bijection when restricted to the set
  $A\cup R$.  Let $a\in A$ and $\mathbf r \in R$. By definition we
  know that $\overline{\varphi}$ maps $\mathbb{P}_{(a,\mathbf r)}$
  homomorphically onto $\mathbb{P}_{(\varphi(a),\varphi^{(k)}(\mathbf
    r))}$. Since $\varphi$ has an inverse $\varphi^{-1}$, it follows
  that $\overline{\varphi^{-1}}$ maps
  $\mathbb{P}_{(\varphi(a),\varphi^{(k)}(\mathbf r))}$ homomorphically
  onto $\mathbb{P}_{(a,\mathbf r)}$.  Thus $\mathbb{P}_{(a,\mathbf
    r)}$ and $\mathbb{P}_{(\varphi(a),\varphi^{(k)}(\mathbf r))}$ are
  isomorphic, completing the proof.
\end{proof}

Using similar arguments it is not hard to prove a bit more, namely that the \emph{monoids of endomorphisms} of $\mathbb A$ and 
$\mathcal D(\mathbb A)$ are isomorphic. Since endormorphisms are just the unary part
of the algebra of polymorphisms, this section can be viewed as a ``baby case''
to the more involved proof in the next section.

\section{Preserving Maltsev conditions}\label{sec:Maltsev}
\noindent Given a finite relational structure $\mathbb A$, we are
interested in the following question: How similar are the algebras of
polymorphisms of $\mathbb A$ and $\mathcal D(\mathbb A)$? More
precisely, which equational properties (or \emph{Maltsev conditions})
do they share? In this section we provide a quite broad range of
Maltsev conditions that hold equivalently in $\mathbb A$ and $\mathcal
D(\mathbb A)$. Indeed, to date, these include all Maltsev conditions
that are conjectured to divide differing levels of tractability and
hardness, as well as all the main tractable algorithmic classes
(e.g. few subpowers and bounded width).

\subsection{The result}
\noindent We start by an overview and statement of the main result of this
section. Since $\mathbb A$ is pp-definable from the digraph $\mathcal D
(\mathbb A)$ (see Lemma \ref{lem:forwardreduction}), it follows that
$A$ and $R$ are subuniverses of $\mathcal D(\mathbb A)$ and for any
$f\in\Pol\mathcal D (\mathbb A)$, the restriction $f|_A$ is a
polymorphism of $\mathbb A$. Consequently, for any set of identities
$\Sigma$,
$$
\mathcal D (\mathbb A)\models\Sigma\text{ implies that }\mathbb A\models\Sigma.
$$
The theorem below, which is a restatement of Theorem \ref{thm:mainresults} (iii), provides 
a partial converse of the above implication.

\begin{thm}\label{thm:preserved_conditions}
  Let $\mathbb A$ be a finite relational structure. Let $\Sigma$ be a
  linear idempotent set of identities such that the algebra of
  polymorphisms of the zigzag satisfies $\Sigma$ and each identity in
  $\Sigma$ is either balanced or contains at most two variables. Then
$$
\mathcal D(\mathbb A)\models\Sigma\text{ if and only if }\mathbb A\models\Sigma.
$$
\end{thm} 

The following corollary lists some popular properties that can be
expressed as sets of identities satisfying the above
assumptions. Indeed, they include many commonly encountered Maltsev
conditions.

\begin{cor}\label{cor:preserved_conditions}
  Let $\mathbb A$ be a finite relational structure.  Then each of the
  following hold equivalently on \up(the polymorphism algebra of\up) $\mathbb A$ and $\mathcal D(\mathbb
  A)$.
\begin{enumerate}[label=(0\arabic*)]
\item[(1)] Being Taylor or equivalently having a weak near-unanimity \up(WNU\up)
  operation \cite{maroti-mckenzie} or equivalently a cyclic
  operation~\cite{b-k} \up(conjectured to be equivalent to being in P if $\mathbb{A}$ is a core~\cite{b-j-k}\up)\up;
\item[(2)] Congruence join-semidistributivity \up($\text{SD}(\vee)$\up) \up(conjectured to be
  equivalent to NL if $\mathbb{A}$ is a core \cite{lartes}\up);
\item[(3)] \up(\up For $n\geq 3$\up) congruence
  $n$-permutability \up(CnP\up) \up(together with \up{(2)} conjectured to be equivalent
  to L if $\mathbb{A}$ is a core \cite{lartes}\up).
  \item[(4)] Congruence meet-semidistributivity \up($\text{SD}(\wedge)$\up)
  \up(equivalent to bounded width~\cite{BW_journal_version}\up)\up;
\item[(5)] \up(For $k\geq 4$\up) $k$-ary edge operation
  \up(equivalent to few subpowers \cite{BIMMVW}, \cite{IMMVW}\up)\up;
\item[(6)] $k$-ary near-unanimity operation \up(equivalent to strict
  width~\cite{fedvar}\up)\up;
\item[(7)] Totally symmetric idempotent \up(TSI\up) operations of all arities
  \up(equivalent to width $1$~\cite{width1}, \cite{fedvar}\up)\up;
\item[(8)] Hobby-McKenzie operations \up(equivalent to the
  corresponding variety satisfying a non-trivial congruence lattice
  identity\up)\up;
\item[(9)] Congruence modularity \up(CM\up);
\item[(10)] Congruence distributivity \up(CD\up);

\end{enumerate}
\end{cor}
\noindent Items (2) and (3) above, together with Theorem \ref{thm:mainresults} (i) and (ii),
show that the Finer CSP complexity conjectures need only be
established in the case of digraphs to obtain a resolution in the
general case.

Note that the above list includes all six conditions for omitting
types in the sense of Tame Congruence Theory~\cite{hobbymckenzie}.
Figure \ref{fig:maltsev}, taken from \cite{JKN13}, presents a diagram
of what might be called the ``universal algebraic geography of CSPs''.
\begin{figure}[h]
  \begin{center}
    \begin{small}
      \xymatrix{ &&& *++[F-,]\txt{Taylor}\ar@{-}[rd] &
        \\
        &&*++[F-,]{\txt{SD$(\wedge)$}}\ar@{-}[ur] &
        &*++[F-,]{\txt{Hobby- \\ McKenzie}} &
        \\
        &&&*++[F-,]{\txt{SD$(\vee)$}}\ar@{-}[ur]\ar@{-}[ul] &
        *++[F-,]\txt{CM}\ar@{-}[u] & *++[F-,]\txt{CnP}\ar@{-}[ul] &
        \\
        &&&*++[F-,]\txt{CD}\ar@{-}[u]\ar@{-}[ur] &
        *++[F-,]{\txt{SD$(\wedge)$ \\ and CnP}}\ar@{-}[ul]\ar@{-}[ur]
        & *++[F-,]\txt{CnP \\ and CM}\ar@{-}[u]\ar@{-}[ul] &
        \\
        &&& & *++[F-,]\txt{CD \\ and
          CnP}\ar@{-}[u]\ar@{-}[ur]\ar@{-}[ul] &
        *++[F-,]\txt{C3P}\ar@{-}[u]
        \\
        &&& & *++[F-,]\txt{CD \\ and C3P}\ar@{-}[u]\ar@{-}[ur] &
        *++[F-,]\txt{Maltsev}\ar@{-}[u]
        \\
        &&& & *++[F-,]\txt{CD \\ and Maltsev}\ar@{-}[u]\ar@{-}[ur] & }
    \end{small}
  \end{center}\caption{The universal algebraic geography of tractable
CSPs.}\label{fig:maltsev}
\end{figure}
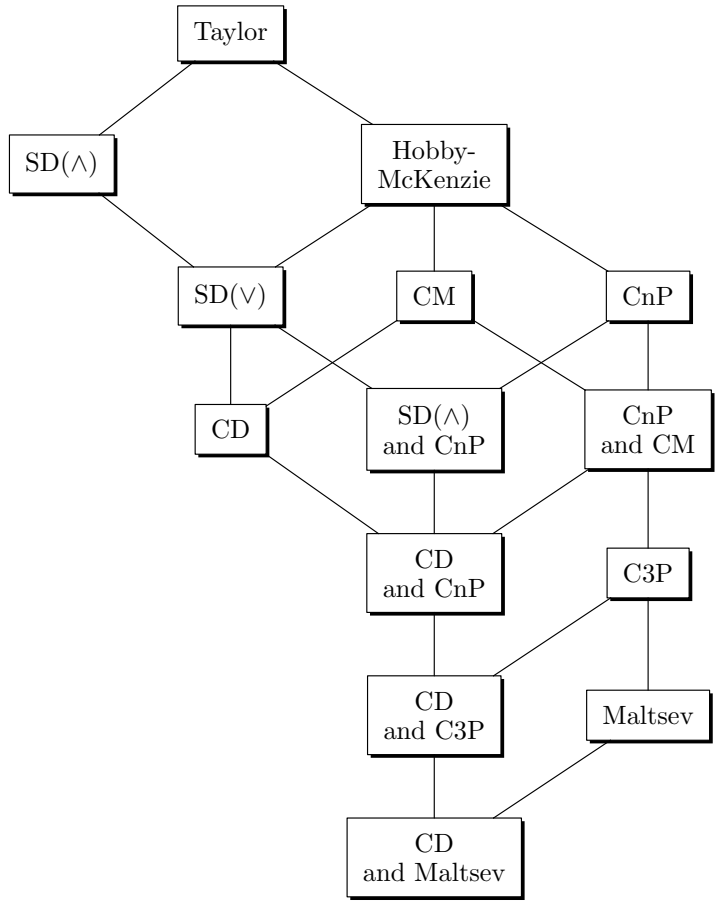

We will prove Theorem \ref{thm:preserved_conditions} and Corollary
\ref{cor:preserved_conditions} in subsection \ref{subsection:proofs}.

\subsection{Polymorphisms of the zigzag}
\noindent In the following, let $\mathbb Z$ be a zigzag with vertices
$00$, $01$, $10$ and $11$ (i.e., the oriented path
\mbox{$00\boldsymbol\rightarrow 01\boldsymbol\leftarrow
  10\boldsymbol\rightarrow 11$}). Let us denote by $\leq_\mathbb Z$ the linear order on $\mathbb Z$ given by $00<_\mathbb Z 01<_\mathbb Z 10<_\mathbb Z 11$. 
  
  Note that the subset $\{00,10\}$ is closed
under all polymorphisms of $\mathbb Z$ (as it is pp-definable using
the formula $(\exists y)(x\boldsymbol\rightarrow y)$, see Lemma
\ref{lem:ppdef_polymorphisms}). The same holds for $\{01,11\}$. We will use
this fact later in our proof. 
  
The digraph $\mathbb Z$ satisfies
most of the important Maltsev conditions (an exception being congruence 2-permutability, i.e., having a Maltsev polymorphism). We need the following.

\begin{lem} \label{lemma:zigzag}
The following holds.
\begin{enumerate}[label=(\roman*)]
\item $\mathbb Z$ has a majority polymorphism,
\item $\mathbb Z$ satisfies any balanced set of identities,
\item $\mathbb Z$ is congruence 3-permutable.
\end{enumerate}
\end{lem}
\begin{proof}
Let $x\wedge y$ and $x\vee y$ denote the binary operations of minimum and maximum with respect to $\leq_\mathbb Z$, respectively. That is, $x\wedge y$ is the vertex from $\{x,y\}$ closer to $00$ and $x\vee y$ the vertex closer to $11$. It can be easily seen that $\wedge,\vee$ are polymorphisms of $\mathbb Z$ and form a distributive lattice. Note that it follows that $\mathbb Z$ satisfies any set of identities which holds in the variety of distributive lattices (equivalently, in the two-element lattice). 

In particular, to prove (i), note that the ternary operation defined by $m(x,y,z)=(x\wedge y)\vee(x\wedge
  z)\vee(y\wedge z)$ (the \emph{median}) is a majority polymorphism. To prove (ii), let $\Sigma$ be a balanced set of identities. For every operation symbol $f$ (say $k$-ary) occurring in $\Sigma$, we define $f^\mathbb Z(x_1,\dots,x_k)=\bigwedge_{i=1}^k x_i$. It is easy to check that $f^\mathbb Z$ is a polymorphism and that such a construction satisfies any balanced identity.
  
To prove (iii), we directly construct the ternary polymorphisms $p_1$ and $p_2$ witnessing $3$-permutability:
 \begin{align*}
 p_1(x,y,z)=&\begin{cases}
 01 & \text{if $y\neq z$ and $01\in\{ x,y,z\}$},\\
 10 & \text{if $y\neq z$ and $10\in\{x,y,z\}$ and $01\notin\{x,y,z\}$},\\
 x & \text{otherwise},
 \end{cases}\\
 p_2(x,y,z)=&\begin{cases}
 01 & \text{if $x\neq y$ and $01\in\{ x,y,z\}$},\\
 10 & \text{if $x\neq y$ and $10\in\{x,y,z\}$ and $01\notin\{x,y,z\}$},\\
 z & \text{if $x=y$}.\\
 x & \text{otherwise}
 \end{cases}
 \end{align*}
The identities $p_1(x,y,y)\approx x$ and $p_2(x,x,y)\approx y$ follow directly from the construction. To verify $p_1(x,x,y)\approx p_2(x,y,y)$ we can assume that $x\neq
 y$. If $01$ or $10$ are in $\{x,y\}$, then $p_1$ and $p_2$ agree (the result is $01$ if $01\in\{x,y\}$ and $10$ else). If not, then $p_1(x,x,y)=p_2(x,y,y)=x$.   
 
Finally, we prove that $p_1$ is a polymorphism of $\mathbb Z$; a similar argument works for $p_2$. If we have triples $\mathbf a,\mathbf b\in Z^3$ such that $a_i\rightarrow b_i$, for $i=1,2,3$, then
 $\{a_1,a_2,a_3\}\subseteq\{00,10\}$ and $\{b_1,b_2,b_3\}\subseteq\{01,11\}$. Thus also $p_1(\mathbf a)\in\{00,10\}$ and $p_1(\mathbf b)\in\{01,11\}$. If $p_1(\mathbf a)=10$, then $p_1(\mathbf a)\rightarrow p_1(\mathbf b)$ follows immediately. If $p_1(\mathbf a)=00$, then $\mathbf a=(00,10,10)$ or $\mathbf a=(00,00,00)$. In both cases $b_1=01$ which gives $p_1(\mathbf b)=01$ and $p_1(\mathbf a)\rightarrow p_1(\mathbf b)$.
\end{proof}

\subsection{The proof} \label{subsection:proofs} 
\noindent In this subsection we prove Theorem
\ref{thm:preserved_conditions} and Corollary
\ref{cor:preserved_conditions}. 
Fix a finite relational structure;
without loss of generality we can assume that $\mathbb A=(A;R)$, where
$R$ is a $k$-ary relation (see Example \ref{example:1}). 

First we need to gather a few facts about connected components of powers of $\mathcal{D}(\mathbb A)$. This is because
 when constructing an $m$-ary polymorphism, one can define it independently on different connected components of $\mathcal D(\mathbb A)^m$ without violating the polymorphism condition. 

We start with the diagonal component: since $\mathcal{D}(\mathbb A)$ is connected, it follows that for every $m>0$ the \emph{diagonal} (i.e., the set $\{(c,c,\dots,c)\mid c\in\mathcal D(\mathbb A)\}$) is connected in $\mathcal D(\mathbb A)^m$. We denote by $\Delta_m$ the connected component of $\mathcal D(\mathbb A)^m$ containing the diagonal.

\begin{lem}\label{obs:2}
 For every $m>0$, both $A^m\subseteq\Delta_m$ and $R^m\subseteq\Delta_m$.
\end{lem}
\begin{proof}
Fix an arbitrary element $a\in A$. Let $(\mathbf r^1,\dots,\mathbf r^m)\in
  R^m$ and for every $i\in[m]$ let $\varphi_i:\mathbb Q\to\mathbb
  P_{(a,\mathbf r^i)}$. The homomorphism defined by $x\mapsto
  (\varphi_1(x),\dots,\varphi_m(x))$ witnesses
  $(a,\dots,a)\stackrel{\mathbb Q}{\longrightarrow}(\mathbf
  r^1,\dots,\mathbf r^m)$ in $\mathcal D(\mathbb A)^m$. This proves that $R^m\subseteq\Delta_m$; a similar argument gives $A^m\subseteq\Delta_m$.
\end{proof}

The next lemma shows that there is only one non-trivial connected
component of $\mathcal{D}(\mathbb{A})^m$ that contains tuples (whose
entries are) on the same level in $\mathcal{D}(\mathbb{A})$; namely
$\Delta_m$. All other such components are singleton.

\begin{lem}\label{lem:diagonalcomponent}
  Let $m>0$ and let $\Gamma$ be a connected component of $\mathcal
  D(\mathbb A)^m$ containing an element $\mathbf c$ such that
  $\lvl(c_1)=\dots=\lvl(c_m)$. Then every element $\mathbf d \in
  \Gamma$ is of the form $\lvl(d_1)=\dots = \lvl (d_m)$ and the
  following hold.
\begin{enumerate}[label=(\roman*)]
\item If $\mathbf c\to\mathbf d$ is an edge in $\Gamma$ such that
  $\mathbf c\notin A^m$ and $\mathbf d\notin R^m$, then there exist
  $e_1,\dots,e_m\in A\times R$ and $l\in[k]$ such that $\mathbf
  c,\mathbf d\in\prod_{i=1}^m\mathbb P_{e_i,l}$.
\item Either $\Gamma=\Delta_m$ or $\Gamma$ is one-element.
\end{enumerate}
\end{lem}
\begin{proof}
  First observe that if an element $\mathbf d$ is connected in
  $\mathcal{D}(\mathbb A)^m$ to an element $\mathbf c$ with $\lvl
  (c_1)= \dots = \lvl (c_m)$, then there is an oriented path
  $\mathbb{Q}'$ such that $\mathbf c \stackrel{\mathbb Q'}\to \mathbf
  d$ from which it follows that $\lvl (d_1)= \dots = \lvl (d_m)$. To prove (i), let $\mathbf c\to\mathbf d$ be an edge in $\Gamma$
  such that $\mathbf c\notin A^m$ and $\mathbf d\notin R^m$. For
  $i=1,\dots, m$ let $e_i$ be such that $c_i\in\mathbb P_{e_i}$ and
  let $l=\lvl (c_1)$. The claim now follows immediately from the
  construction of $\mathcal D(\mathbb A)$.
  
  It remains to prove (ii).  If $|\Gamma|>1$, then there is an edge
  $\mathbf c\to\mathbf d$ in $\Gamma$. If $\mathbf c\in A^m$ or
  $\mathbf d\in R^m$, then the claim follows from
  Lemma~\ref{obs:2}. Otherwise, from~(i), there exists $l\in[k]$ and
  $e_i=(a_i,\mathbf r^i)$ such that $\mathbf c,\mathbf
  d\in\prod_{i=1}^m\mathbb P_{e_i,l}$. For every $i\in[m]$ we can walk from $c_i$ to $\iota\mathbb
  P_{e_i,l}$ following the path
  \mbox{$\bullet\boldsymbol\rightarrow\bullet\boldsymbol\leftarrow\bullet$};
  and so $\mathbf c$ and $(\iota\mathbb P_{e_1,l},\dots,\iota\mathbb
  P_{e_m,l})$ are connected. For every $i\in[m]$ there exists a
  homomorphism $\varphi_i:\mathbb Q\to\mathbb P_{e_i}$ such that
  $\varphi_i(\iota\mathbb Q)=a_i$ and $\varphi_i(\iota\mathbb
  Q_l)=\iota\mathbb P_{e_i,l}$. The homomorphism $\mathbb Q\to
  \mathcal D(\mathbb A)^m$ defined by $x\mapsto
  (\varphi_1(x),\dots,\varphi_m(x))$ shows that $(a_1,\dots,a_m)$ and
  $(\iota\mathbb P_{e_1,l},\dots,\iota\mathbb P_{e_m,l})$ are
  connected. By transitivity, $(a_1,\dots,a_m)$ is connected to
  $\mathbf c$ and therefore $(a_1,\dots,a_m)\in \Gamma$. Using (i) we
  obtain $\Gamma=\Delta_m$.
\end{proof}

In order to deal with connected components that contain tuples of varying levels, we need to define two linear orders $\sqsubseteq,\sqsubseteq^\star$ on $\mathcal D(\mathbb A)$. These linear orders will then be used to choose elements from input tuples of the polymorphisms under construction in a ``uniform'' way.

Fix an arbitrary linear order $\preceq$ on $A$. It induces lexicographic orders on relations on $A$. We will use $\preceq_\mathrm{LEX}$ on $R$, $A\times R$ and also on $R\times A$. (Note the difference!) We define the linear order $\sqsubseteq$ on $\mathcal D(\mathbb A)$ by putting $x\sqsubset y$ if any of the following five conditions holds:
\begin{enumerate}
\item[(1)] $x,y\in A$ and $x\prec y$, or 
\item[(2)] $x,y\in R$ and $x\prec_\mathrm{LEX} y$, or
\item[(3)] $\lvl(x)<\lvl(y)$,
\end{enumerate}
or $\lvl(x)=\lvl(y)$, $x,y\notin A\cup R$, say $x\in\mathbb P_{(a,\mathbf r)}$, $y\in\mathbb P_{(b,\mathbf s)}$, and
\begin{enumerate}
\item[(4)] $(a,\mathbf r)=(b,\mathbf s)$ and $x$ is closer
  to $\iota\mathbb P_{(a,\mathbf r)}$ than $y$, or
\item[(5)] $(a,\mathbf r)\prec_\mathrm{LEX}(b,\mathbf s)$.
\end{enumerate}
We also define the linear order $\sqsubseteq^\star$, which will serve as a ``dual'' to $\sqsubseteq$ in some sense. The definition is almost identical, we put $x\sqsubset^\star y$ if one of (1), (2), (3), (4) or ($5^\star$) holds, where
\begin{enumerate}[label=(\arabic*$^\star$)]
\item[($5^\star$)] $(\mathbf r,a)\prec_\mathrm{LEX}(\mathbf s,b)$.
\end{enumerate}
\noindent The last ingredient is the following lemma; $\sqsubseteq$ and $\sqsubseteq^\star$ were tailored to satisfy it.

\begin{lem} \label{lemma:the_orders}
  Let $C$ and $D$ be subsets of $\mathcal D(\mathbb A)$ such that
\begin{itemize}
\item for every $x\in C$ there exists $y'\in D$ such that
  $x\boldsymbol{\rightarrow}y'$, and
\item for every $y\in D$ there exists $x'\in C$ such that
  $x'\boldsymbol{\rightarrow}y$.
\end{itemize}
Then the following is true.
\begin{enumerate}[label=(\roman*)]
\item If $D\nsubseteq R$ and $c$ and $d$ are the $\sqsubseteq$-minimal 
elements of $C$ and $D$,
respectively, then $c\boldsymbol{\rightarrow}d$.
\item If $C\nsubseteq A$ and $c$ and $d$ are the $\sqsubseteq^\star$-maximal 
elements of $C$ and $D$,
respectively, then $c\boldsymbol{\rightarrow}d$.
\end{enumerate}
\end{lem}
\begin{proof}
We will prove item (ii); the proof of (i) is similar. Let $c',d'$ be such that $c\boldsymbol{\rightarrow}d'$ and $c'\boldsymbol{\rightarrow}d$. There exist $(a,\mathbf r),(b,\mathbf s)\in A\times R$ such that $c,d'\in\mathbb P_{(a,\mathbf r)}$ and $c',d\in\mathbb P_{(b,\mathbf s)}$. Suppose for contradiction that $c\not\boldsymbol{\rightarrow}d$. In particular, $c\neq c'$ and $d\neq d'$. Note that 
the assumptions of $c,c',d,d'$ and 
item (3) of the definition of $\sqsubseteq^\star$ give
$\lvl(c')+1=\lvl(d)\geq \lvl(d')=\lvl(c)+1\geq \lvl(c')+1$, so that 
$\lvl(c)=\lvl(c')$ and  $\lvl(d)=\lvl(d')$. 
So, the reason for $d'\sqsubset^\star d$ must be one of items (2), (4) or~($5^\star$).

If it is (2), then $d'=\mathbf r$ and $d=\mathbf s$ with $\mathbf r\prec_{LEX}\mathbf s$.  Therefore $(\mathbf r,a)\prec_\mathrm{LEX}(\mathbf s,b)$ and ($5^\star$) gives us $c\sqsubset^\star c'$, a contradiction with the maximality of $c$. If it is (4), then $(a,\mathbf r)=(b,\mathbf s)$ and $c\boldsymbol{\rightarrow}d'\boldsymbol{\leftarrow}c'\boldsymbol{\rightarrow}d$ form a zigzag. By (4) we again get $c\sqsubset^\star c'$. In case the reason for $d'\sqsubset^\star d$ is ($5^\star$), the same item gives $c\sqsubset^\star c'$. (Here we need the assumption that $C\nsubseteq A$, otherwise we could have $c=a$, $c'=b$, $b\prec a$ and 
$c'\sqsubset^\star c$ by (1) even though $(\mathbf r,a)\prec_\mathrm{LEX}(\mathbf s,b)$.) 
\end{proof}

\subsubsection*{Proof of Theorem \ref{thm:preserved_conditions}.}
\noindent Let $\Sigma$ be a set of identities in operation symbols
$\{f_\lambda:\lambda\in\Lambda\}$ satisfying the assumptions. Let
$\{f^\mathbb A_\lambda\mid\lambda\in\Lambda\}$ and $\{f^\mathbb
Z_\lambda\mid\lambda\in\Lambda\}$ be interpretations of the operation
symbols witnessing $\mathbb A\models\Sigma$ and $\mathbb
Z\models\Sigma$, respectively.

We will now define polymorphisms $\{f^{\mathcal D(\mathbb
  A)}_\lambda\mid\lambda\in\Lambda\}$ witnessing that $\mathcal
D(\mathbb A)\models\Sigma$. Fix $\lambda\in\Lambda$ and assume that
$f_\lambda$ is $m$-ary. We split the definition of $f^{\mathcal
  D(\mathbb A)}_\lambda$ into several cases and subcases. Let $\mathbf
c\in \mathcal D(\mathbb A)^m$ be an input tuple.

\medskip

\noindent\textbf{Case 1.} $\mathbf c\in A^m\cup R^m$.

\smallskip

\noindent\underline{1a}\quad If $\mathbf c\in A^m$, we define
$f^{\mathcal D(\mathbb A)}_\lambda(\mathbf c)=f^\mathbb
A_\lambda(\mathbf c)$.

\smallskip

\noindent\underline{1b}\quad If $\mathbf c\in R^m$, we define
$f^{\mathcal D(\mathbb A)}_\lambda(\mathbf c)=(f^\mathbb
A_\lambda)^{(k)}(\mathbf c)$.

\smallskip

\noindent\textbf{Case 2.} $\mathbf c\in\Delta_m\setminus(A^m\cup R^m)$.

\smallskip

\noindent Let $c_i\in\mathbb P_{e_i}$ and define $e=(f^\mathbb
A_\lambda)^{(k+1)}(e_1,\dots,e_m)$. Let $l\in[k]$ be minimal such that
$c_i\in\mathbb P_{e_i,l}$ for all $i\in[m]$. (Its existence is
guaranteed by Lemma \ref{lem:diagonalcomponent} (i).)

\smallskip

\noindent\underline{2a}\quad If $\mathbb P_{e,l}$ is a single edge,
then we define $f^{\mathcal D(\mathbb A)}_\lambda(\mathbf c)$ to be
the vertex from $\mathbb P_{e,l}$ having the same level as all the
$c_i$'s.

\smallskip 

\noindent If $\mathbb P_{e,l}$ is a zigzag, then at least one of the
$\mathbb P_{e_i,l}$'s is a zigzag as well. (This follows from the
construction of $\mathcal D(\mathbb A)$ and the fact that $f^\mathbb A_\lambda$ preserves $R$.)
For every $i\in[m]$ such that $\mathbb P_{e_i,l}$ is a zigzag let
$\Phi_i:\mathbb P_{e_i,l}\to\mathbb Z$ be the (unique)
isomorphism. Let $\Phi$ denote the isomorphism from $\mathbb P_{e,l}$
to $\mathbb Z$.

\smallskip

\noindent\underline{2b}\quad If all of the $\mathbb P_{e_i,l}$'s are
zigzags, then the value of $f^{\mathcal D(\mathbb A)}_\lambda$ is
defined as follows:
$$
f^{\mathcal D(\mathbb A)}_\lambda(\mathbf c)=\Phi^{-1}(f^\mathbb
Z_\lambda(\Phi_1(c_1),\dots,\Phi_m(c_m))).
$$

\smallskip

\noindent\underline{2c}\quad Otherwise, we define $f^{\mathcal D(\mathbb
  A)}_\lambda(\mathbf c)$ to be the $\sqsubseteq$-minimal element from the set 
$$
\{\Phi^{-1}(\Phi_i(c_i))\mid\mathbb P_{e_i,l}\text{ is a zigzag}\}.
$$
(Equivalently, $f^{\mathcal D(\mathbb
  A)}_\lambda(\mathbf c)=\Phi^{-1}(z)$, where $z$ is the $\leq_\mathbb Z$-minimal element from 
 $\{\Phi_i(c_i)\mid\mathbb P_{e_i,l}\text{ is a zigzag}\}$.)

\smallskip

\noindent\textbf{Case 3.} $\mathbf c\notin\Delta_m$.

\smallskip

\noindent\underline{3a}\quad If $|\{\lvl(c_i)\mid i\in[m]\}|=1$ and the $c_i$'s lie on precisely two paths (say, $\{c_1,\dots,c_m\}\subseteq\mathbb P_e\cup\mathbb P_{e'}$ with $e\prec_{LEX} e'$, the lexicographic order of $A\times R$), 
then we define the mapping $\Psi:\{c_1,\dots,c_m\}\to\{00,10\}$ as follows:
$$
\Psi(c_i)=\begin{cases}
00 & \text{if }c_i\in\mathbb P_e,\\
10 & \text{if }c_i\in\mathbb P_{e'}.
\end{cases}
$$
We define $f^{\mathcal D(\mathbb A)}_\lambda(\mathbf c)$ to be the $\sqsubseteq$-minimal element from the set 
$$
\{c_i:\Psi(c_i)=f^\mathbb Z_\lambda(\Psi(c_1),\dots,\Psi(c_m))\}.
$$

\smallskip

\noindent\underline{3b}\quad If $|\{\lvl(c_i)\mid i\in[m]\}|=2$ 
(say, $\lvl(c_i)\in\{l,l'\}$ for all $i\in[m]$ and $l<l'$), 
then we define the mapping $\Theta:\{c_1,\dots,c_m\}\to\{00,10\}$ as follows:
$$
\Theta(c_i)=\begin{cases}
00 & \text{if }\lvl(c_i)=l\\
10 & \text{if }\lvl(c_i)=l'.
\end{cases}
$$
We set $z=f^\mathbb Z_\lambda(\Theta(c_1),\dots,\Theta(c_m))$ and $C'=\{c_i:\Theta(c_i)=z\}$ and define
$$
f^{\mathcal D(\mathbb A)}_\lambda(\mathbf c)=\begin{cases}
\text{the $\sqsubseteq$-minimal element from }C' & \text{if }z=00\\
\text{the $\sqsubseteq^\star$-maximal element from }C' & \text{if }z=10.
\end{cases}
$$

\smallskip

\noindent\underline{3c}\quad In all other cases we define $f^{\mathcal
  D(\mathbb A)}_\lambda(\mathbf c)$ to be the $\sqsubseteq$-minimal
element from the set $\{c_1,\dots,c_m\}$.

\medskip

While the construction is a bit technical, the ideas behind it are not so complicated. Case 1 gives us no choice. In Case 2 we use $f^\mathbb A$ to determine on which path $\mathbb P_e$ should the result lie, and we are left with a choice of at most two possible elements (when $\mathbb P_{e,l}$ is a zigzag). In Case 3 we cannot use $f^\mathbb A$ anymore. Instead, we choose the result as a minimal element from (a subset of) the input elements under a suitable linear order $\sqsubseteq$. This choice typically does not depend on order or repetition of the input elements, which allows us to satisfy balanced identities ``for free''. The trickiest part is to deal with connected components which can contain tuples with just two distinct elements, as these can play a role in some non-balanced identity (in two variables) which we need to satisfy. We need to employ $f^\mathbb Z$ to choose from two possibilities: a result which is the right element (in subcase 2c), from the right path (in 3a) or from the right level (in 3b). We then use $\sqsubseteq$ to choose the result from the ``good'' elements (and as a technical nuisance, to maintain the polymorphism property, in 3b we sometimes need to use $\sqsubseteq^\star$-maximal elements instead).

We need to verify that the operations we constructed are polymorphisms
and that they satisfy all identities from $\Sigma$. We divide the
proof into three claims.

\begin{clm}
  For every $\lambda\in\Lambda$, $f^{\mathcal D(\mathbb A)}_\lambda$
  is a polymorphism of $\mathcal D(\mathbb A)$.
\end{clm}
\begin{proof}
  Let $\mathbf c\boldsymbol\rightarrow\mathbf d$ be an edge in
  $\mathcal D(\mathbb A)^m$. Note that $\mathbf c\in\Delta_m$ if and
  only if $\mathbf d\in\Delta_m$. The tuple $\mathbf c$ cannot fall
  under subcase \underline{1b} or under \underline{3a}, because these
  cases both prevent an outgoing edge from $\mathbf c$ (see Lemma
  \ref{lem:diagonalcomponent} (ii) for why this is true for
  \underline{3a}).

  We first consider the situation where $\mathbf c$ falls under
  subcase \underline{1a} of the definition. Then $\mathbf d$ falls
  under case 2 and, moreover, $d_i=\iota\mathbb P_{e_i,1}$ for all
  $i\in[m]$. It is not hard to verify that $f^{\mathcal D(\mathbb
    A)}_\lambda(\mathbf d)=\iota\mathbb P_{e,1}$. (In subcase
  \underline{2b} we need the fact that $f^\mathbb Z_\lambda$ is
  idempotent.) Therefore $f^{\mathcal D(\mathbb A)}_\lambda(\mathbf
  c)=\iota\mathbb P_e\boldsymbol{\rightarrow}\iota\mathbb
  P_{e,1}=f^{\mathcal D(\mathbb A)}_\lambda(\mathbf d)$ and the
  polymorphism condition holds. The argument is similar when $\mathbf
  d$ falls under subcase \underline{1b} (and so $\mathbf c$ under case
  2).

  Consider now that $\mathbf c$ falls under case 2. Then $\mathbf d$
  falls either under subcase \underline{1b}, which was handled in the
  above paragraph, or also under case 2. The elements $e_1,\dots,e_m$
  and $e$ are the same for both $\mathbf c$ and $\mathbf d$. By Lemma
  \ref{lem:diagonalcomponent} (i), there exists $l\in[k]$ such that
  $c_i,d_i\in\mathbb P_{e_i,l}$ for all $i\in[m]$.

  If the value of $l$ is also the same for both $\mathbf c$ and
  $\mathbf d$, then $f^{\mathcal D(\mathbb A)}_\lambda(\mathbf
  c)\boldsymbol{\rightarrow}f^{\mathcal D(\mathbb A)}_\lambda(\mathbf
  d)$ follows easily; in subcase \underline{2a} trivially, in
  \underline{2b} from the fact that $f^\mathbb Z_\lambda$ is a
  polymorphism of $\mathbb Z$ and in \underline{2c} from Lemma \ref{lemma:the_orders}.

  It may be the case that this $l$ is not minimal for the tuple
  $\mathbf c$, that is, that $c_i\in\mathbb P_{e_i,l-1}\cap\mathbb P_{e_i,l}$ for all
  $i\in[m]$. But then $c_i=\tau\mathbb
  P_{e_i,l-1}=\iota\mathbb P_{e_i,l}$ and thus $f^{\mathcal D(\mathbb
    A)}_\lambda(\mathbf c)=\iota\mathbb P_{e_i,l}$ (again, using
  idempotency of $f^\mathbb Z_\lambda$ in subcase
  \underline{2b}). Knowing this allows for the same argument as in the
  above paragraph.

  If $\mathbf c$ falls under one of the subcases \underline{3b} or
  \underline{3c}, then $\mathbf d$ falls under the same subcase. In subcase \underline{3c} we apply  \ref{lemma:the_orders} (i) with $\{c_1,\dots,c_m\}$ and $\{d_1,\dots,d_m\}$ in the roles of $C$ and $D$, respectively. In subcase \underline{3b} our construction ``chooses'' either the lower
    or the higher level, and it is easy to see that this choice (i.e., the element $z$) is the
  same for both $\mathbf c$ and $\mathbf d$. We then apply Lemma \ref{lemma:the_orders} (i) or (ii) (depending on $z$, note that the assumptions are satisfied) with $C'=\{c_i:\Theta(c_i)=z\}$ and $D'=\{d_i:\Theta(d_i)=z\}$ in the role of $C$ and $D$, respectively. In both cases we get $f^{\mathcal D(\mathbb A)}_\lambda(\mathbf
    c)\boldsymbol{\rightarrow}f^{\mathcal D(\mathbb A)}_\lambda(\mathbf
    d)$.
\end{proof}

\begin{clm}
  The $f^{\mathcal D(\mathbb A)}_\lambda$'s satisfy every balanced
  identity from $\Sigma$.
\end{clm}
\begin{proof}
  Let $f_\lambda(\mathbf u)\approx f_\mu(\mathbf v)\in\Sigma$ be a
  balanced identity in $s$ distinct variables $\{x_1,\dots,x_s\}$. Let
  $\mathcal E:\{x_1,\dots,x_s\}\to\mathcal D(\mathbb A)$ be some
  evaluation of the variables. Let $\mathbf u^\mathcal E$ and $\mathbf
  v^\mathcal E$ denote the corresponding evaluation of these tuples.

  Note that both $f^{\mathcal D(\mathbb A)}_\lambda(\mathbf u^\mathcal
  E)$ and   $f^{\mathcal D(\mathbb A)}_\mu(\mathbf v^\mathcal E)$ fall under the same subcase of
  the definition. The subcase to be applied depends only on the set of
  elements occuring in the input tuple, except for case two, where the
  choice of $e$ matters as well. However, since the identity
  $f_\lambda(\mathbf u)\approx f_\mu(\mathbf v)$ holds in $\mathbb A$,
  this $e$ is the same for both $\mathbf u^\mathcal E$ and $\mathbf
  v^\mathcal E$. Therefore, to verify that $f^{\mathcal D(\mathbb
    A)}_\lambda(\mathbf u^\mathcal E)=f^{\mathcal D(\mathbb
    A)}_\mu(\mathbf v^\mathcal E)$, it is enough to consider the
  individual subcases separately.

  In case 1 it follows immediately from the fact that the identity
  holds in $\mathbb A$. In case 2 it is easily seen that both
  $f^{\mathcal D(\mathbb A)}_\lambda(\mathbf u^\mathcal E)$ and 
  $f^{\mathcal D(\mathbb A)}_\mu(\mathbf v^\mathcal E)$ have the same level, and since the
  identity holds in $\mathbb A$, they also lie on the same path
  $\mathbb P_{e,l}$. To see that these two elements are equal, note
  that in subcase \underline{2a} it is trivial, in \underline{2b} it
  follows directly from the fact that the identity holds in $\mathbb
  Z$, and in \underline{2c} we use the fact that the identity is
  balanced: they are both the $\sqsubseteq$-minimal element of the same
  set.
  
  Similar arguments can be used in case 3. In \underline{3a} we choose one of the paths $\mathbb P_e$, $\mathbb P_{e'}$; the choice is the same because $f_\lambda(\mathbf u)\approx f_\mu(\mathbf v)$ holds in $\mathbb Z$. Both $f^{\mathcal D(\mathbb A)}_\lambda(\mathbf u^\mathcal E)$ and $f^{\mathcal D(\mathbb A)}_\mu(\mathbf v^\mathcal E)$ then evaluate to the same element, namely the $\sqsubseteq$-minimal element from $\{\mathcal E(x_1),\dots,\mathcal E(x_s)\}$ intersected with the chosen path.  In \underline{3b} the chosen level is the
  same for both of them (since the identity holds in $\mathbb Z$) and
  they are both the $\sqsubseteq$-minimal, or $\sqsubseteq^\star$-maximal, element of the set of elements
  from $\{\mathcal E(x_1),\dots,\mathcal E(x_s)\}$ lying on that
  level. In \underline{3c} both are the $\sqsubseteq$-minimal element of
  the same set $\{\mathcal E(x_1),\dots,\mathcal E(x_s)\}$.
\end{proof}

\begin{clm}
  The $f^{\mathcal D(\mathbb A)}_\lambda$'s satisfy every identity
  from $\Sigma$ in at most two variables.
\end{clm}
\begin{proof}
  Balanced identities fall under the scope of the previous
  claim. Since $\Sigma$ is idempotent, we may without loss of
  generality consider only identities of the form $f_\lambda(\mathbf
  u)\approx x$, where $\mathbf u\in\{x,y\}^m$. Suppose that $x$ and
  $y$ evaluate to $c$ and $d$ in $\mathcal D(\mathbb A)$, respectively,
  and let $\mathbf c\in\{c,d\}^n$ be the corresponding evaluation of $\mathbf u$. We
  want to prove that $f^{\mathcal D(\mathbb A)}_\lambda(\mathbf
  c)=c$.

  The tuple $\mathbf c$ cannot fall into subcase \underline{3c} of the
  definition of $f^{\mathcal D(\mathbb A)}_\lambda$. If it falls into
  case 1, the equality follows from the fact that the identity holds
  in $\mathbb A$, while in subcases \underline{3a} and \underline{3b}
  we use the fact that it holds in $\mathbb Z$. (The linear orders $\sqsubseteq,\sqsubseteq^\star$ do not matter, since we only choose elements from singleton sets.)

  In case 2 it is easily seen that $f^{\mathcal D(\mathbb
    A)}_\lambda(\mathbf c)$ lies on the same path $\mathbb P_{e,l}$ as
  $c$ (using that the identity holds in $\mathbb A$) as well as on the
  same level of this path. In \underline{2a} it is trivial that
  $f^{\mathcal D(\mathbb A)}_\lambda(\mathbf c)=c$ while in
  \underline{2b} it follows from the fact that the identity holds in
  $\mathbb Z$. If $\mathbf c$ falls under subcase \underline{2c}, then $c\in\mathbb P_{e,l}$, which is a zigzag, and $d\in\mathbb P_{e',l}$, which must be a single edge. Therefore $f^{\mathcal D(\mathbb A)}_\lambda(\mathbf c)$ is defined to be the $\sqsubseteq$-minimal element from the singleton set $\{c\}$.
\end{proof}

\subsubsection*{Proof of Corollary \ref{cor:preserved_conditions}.}
\noindent All items are expressible by linear idempotent sets of
identities. In all items except (7) they are in at most two variables,
in item (7) the defining identities are balanced. It remains to check
that all these conditions are satisfied in the zigzag, which follows
from Lemma \ref{lemma:zigzag} (iii) for item (3), Lemma \ref{lemma:zigzag} (ii) for item (7) and Lemma \ref{lemma:zigzag} (i) for all other items.

\section{The logspace reduction of \texorpdfstring{$\CSP(\mathcal{D}(\A))$}{CSP(D(A))} to
  \texorpdfstring{$\CSP(\A)$}{CSP(A)}.}\label{sec:reversereduction}
\noindent In this section, we give the proof of Lemma
\ref{lem:reversereduction}, by showing that $\CSP(\mathcal{D}(\A))$
reduces in logspace to $\CSP(\A)$.  A sketch of a \emph{polymomial
  time} reduction is given in the proof of \cite[Theorem 13]{fedvar};
technically, that argument is for the special case where $\mathbb{A}$
is itself already a digraph, but the arguments can be broadened to
cover our case.  To perform this process in logspace is rather
technical, with many of the difficulties lying in details that are
omitted in the polymomial time description in the proof of
\cite[Theorem 13]{fedvar}.  We wish to thank Barnaby Martin for
encouraging us to pursue Lemma \ref{lem:reversereduction}.

The following theorem is an immediate consequence of Lemma
\ref{lem:forwardreduction} and Lemma~\ref{lem:reversereduction}. As
this improves the oft-mentioned polynomial time equivalence of general
CSPs with digraph CSPs, we now present it as stand-alone statement.
\begin{thm}\label{thm:logspace}
  Every fixed finite template CSP is logspace equivalent to the CSP
  over some finite digraph.
\end{thm}

\subsection{Outline of the algorithm.}
Algorithms running in logspace are often quite technical, so we begin by giving a broad overview of how the algorithm produces an output.
We first assume that $\CSP(\mathbb{A})$ is itself not trivial (that
is, that there is at least one no instance and one yes instance): this
uninteresting restriction is necessary because
$\CSP(\mathcal{D}(\A))$ will have no instances always.  Also, let $n$
denote the height of $\mathcal{D}(\A)$ and $k$ the arity of the single
fundamental relation $R$ of $\A$: so, $n=k+2$. Recall that the
vertices of $\mathcal{D}(\A)$ include those of $A$ as well as the
elements in $R$.

Now
let $\G$ be an instance of $\CSP(\mathcal{D}(\A))$.   The algorithm produces a structure $\mathbb{B}$ of the same signature as $\mathbb{A}$ that is a YES instance of $\CSP(\mathbb{A})$ if and only if $\G$ is a YES instance of $\CSP(\mathcal{D}(\A))$.  We argue below that each component of $\G$ can be considered separately, and so for this reason we will assume here that $\G$ has a single connected component. The rough outline of the algorithm is as follows.

\begin{description}
\item[Stage 1] Test whether or not $\G$ is a balanced digraph of height at most $n$.  If not, some fixed no instance of $\CSP(\A)$
  is output.
\item[Stage 2] If the height of $\G$ is less than $n$, then we directly test if $\G$ is a YES or NO instance of
  $\CSP(\mathcal{D}(\A))$.  If YES then we return some fixed YES instance of $\CSP(\A)$.  If NO then we return some fixed NO
  instance of $\CSP(\A)$.
\item[Stage 3] If Stage 3 is reached then $\G$ is a balanced digraph of height $n$.  The goal is to output a structure $\mathbb B$ such that $\mathcal{D}(\B)$ is a YES instance of $\CSP(\mathcal{D}(\A))$ if and only if $\G$ is a YES instance of $\CSP(\mathcal{D}(\A))$.  The vertices of $\G$ at height $0$ can be thought of as being similar to the height $0$ vertices of $\mathcal{D}(\B)$ (which of course, are the actual vertices of $\B$), and vertices of $\G$ at height $n$ can be thought of as being similar  to the hyperedge vertices of $\mathcal{D}(\B)$.  In general though, the similarity between $\G$ and our desired $\mathcal{D}(\B)$ can be quite ``blurry'', with some distinct vertices of $\G$ corresponding to single vertices of $\mathcal{D}(\B)$, and other vertices of $\mathcal{D}(\B)$ simply missing from $\G$ completely.  
Because of this, we work in
  two steps. 
  \item[Stage 3A]  From $\G$ we output a list of ``generalized hyperedges''.  These are $k$-tuples consisting of sets of vertices of $\G$ plus some newly added vertices.  Moreover, they sometimes include a labelling to record how they were created.
\item[Step 3B] The output structure $\B$ is constructed from the
  generalized hyperedges in 3A.  To create $\B$, numerous undirected graph
  reachability checks are performed.  The final ``vertices'' of $\B$
  are in fact sets of vertices of $\G$, so that the generalized
  hyperedges become actual hyperedges in the conventional sense
  ($k$-tuples of ``vertices'', now consisting of sets of vertices of
  $\G$).  This may be reduced to an adjacency matrix description as a
  separate logspace process, but that is routine.
\end{description}
\medskip

\noindent Stage 1 is described in Subsection~\ref{subsec:1}, while Stage 2 is
described in Subsection~\ref{subsec:2}.  The most involved part of the
algorithm is Stage 3A.  In Subsection \ref{subsec:3} we give some
preliminary discussion on how the process is to proceed: an
elaboration on the item listed in the present subsection.  In
particular a number of definitions are introduced to aid the
description of Stage 3A.  \ The actual subroutine for Stage~3A is detailed in
Subsection~\ref{subsec:3A}.  Step~3B is described in Subsection~\ref{subsec:3B}.  
After a brief discussion of why the algorithm is a
valid reduction from $\CSP(\mathcal{D}(\A))$ to $\CSP(\A)$, we present
an example of Stages 3A and 3B in action.  This example may be a
useful reference while reading Subsection~\ref{subsec:3A} and
\ref{subsec:3B}.

Before we begin describing the algorithm we recall some basic logspace
processes that we will use frequently.
\subsection{Subroutines}\label{subsec:routines}
The algorithm we describe makes numerous calls on other log\-space
computable processes; more formally, it can be implemented on
an oracle machine with several query tapes each an oracle for some known logspace problem. It is well known that
$\texttt{L}^\texttt{L}=\texttt{L}$ (logspace on machine with logspace oracle is the same as logspace), and this enables all of the query
tapes to be eliminated within logspace.  For the sake of clarity, we
briefly recall some basic information on logspace on an oracle
machine.  An oracle program with logspace query language $U$ has
access to an input tape, a working tape (or tapes), an output tape and a query tape.
Unlimited reading may be performed on the input, but no writing.
Unlimited writing may be performed on the output tape, but no reading.
Unlimited writing may be done to the query tape, but no reading, \emph{except} that once
the query state is reached, the current word written to the
query tape is tested for membership in the language $U$ (at the cost
of one step of computation), and a (correct) answer of either yes or
no is received by the program.  The query tape is then immediately
erased at no further cost.  The space used is measured \emph{only from the working tape},
where both reading and writing is allowed.  If such a program runs in
logspace, then it can be emulated by an actual logspace program (with
no oracle).  The argument
is essentially the fact that a composition of logspace reductions is a
logspace reduction: each query to the oracle (of a string $w$ for
instance) during the computation is treated as a fresh instance of a
reduction to the membership problem of $U$, which is then composed
with the logspace algorithm for $U$.  As usual for composition of logspace programs, this would be performed without
ever writing any more than around one symbol of $w$ at a time, plus a counter recording the bit position (see  Papadimitriou \cite{1994-papadimitriou} for example).  This is why space used on the oracle tape does not
matter in the oracle formulation of logspace, and why we may assume
that the query tape may be erased after completion of the query.

We now describe the basic checks that are employed
during our algorithm.

\begin{description}
\item[Undirected reachability] Given an undirected graph and two
vertices $u,v$, there is a logspace algorithm to determine if $u$ is
reachable from $v$ (Reingold \cite{DBLP:conf/stoc/Reingold05,rei08}).  In the case of a
directed graph we may use this to determine if two vertices are
connected by some oriented path: it is simply undirected reachability in the underlying graph.  As an example, consider a binary relation $\beta$ on a finite set $X$ with the property that recognising membership of pairs in $\beta$ can be decided in logspace.  Let $\bar\beta$ be the smallest equivalence relation containing $\beta$.  Now, membership of a pair in $\bar\beta$ is simply connectivity in the underlying graph of $\beta$.  Hence, given any $x\in X$, we may in logspace output the lexicographically earliest vertex from the equivalence class of $x$ modulo $\bar\beta$: simply search through $x$ testing for $\bar\beta$-relatedness with $x$, and output the first vertex that returns a positive answer.

A second process we frequently perform is reachability checks 
involving edges that are not precisely those of the current input digraph. 
A typical instance might be where we have some fixed
vertex $u$ in consideration, and we wish to test if some vertex $v$
can be reached from $u$ by an oriented path consisting only of vertices satisfying some property $\mathcal{Q}$, where $\mathcal{Q}$ is a logspace
testable property.  This is undirected graph reachability,
except that as well as ignoring the edge direction, we must also ignore any
vertex failing property $\mathcal{Q}$.  This can be performed in logspace 
on an oracle machine running
an algorithm for undirected graph reachability and whose query tape
tests property $\mathcal{Q}$.

\item[Component checking] Undirected graph reachability is also
fundamental to checking properties of induced subgraphs.  In a typical
situation we have some induced subgraph $C$ of $\G$ (containing some
vertex $u$, say) and we want to test if it satisfies some property
$\mathcal{P}$.  Membership of vertices in $C$ is itself determined by
some property $\mathcal{Q}$, testable in logspace.  It is convenient
to assume that the query tape for $\mathcal{P}$ expects inputs that
consist of a list of directed edges.  We may construct a list of the directed edges
in the component $C$ on a logspace machine with query tapes for
$\mathcal{P}$, for undirected graph reachability and for
$\mathcal{Q}$.  We write $C$ to the query tape for $\mathcal{P}$ as
follows.  Systematically enumerate pairs of vertices $v_1,v_2$ of $\G$
(re-using some fixed portion of work tape for each pair), in each case
testing for undirected reachability of both $v_1$ and $v_2$ from $u$,
and also for satisfaction of property $\mathcal{Q}$.  If both are
reachable, and if $(v_1,v_2)$ is an edge of $\G$ then we output the
edge $(v_1,v_2)$ to the query tape for $\mathcal{P}$.  After the last
pair has been considered, we may finally query~$\mathcal{P}$.

\item[Testing for interpretability in paths] By an \emph{interpretation} of a digraph $C$ in another digraph $\Q$ we mean simply a graph homomorphism from $C$ to $\Q$.
The basic properties we wish to
test of components usually concern interpretability within some fixed finite
family of directed paths.  We consider the paths $\Q_S$, where $S$ is
some subset of $[k]=\{1,\dots,k\}$: recall (Section \ref{sec:reduction}) that
these have zigzags in
a position $i$ when $i\notin S$ (so that a small $S$ corresponds to a large number of zigzags, while $\Q_{[k]}$ itself is simply the directed
path on $k+3$ vertices, with no zigzags).
\end{description}

\noindent It is not hard to see that a balanced digraph of height $n=k+2$ admits a homomorphism into $\mathbb{Q}_S$ if and only if it admits a homomorphism into each of $\mathbb{Q}_{[k]\backslash\{i\}}$ for $i\notin S$ (this is discussed further in the proof of the next lemma).  For balanced digraphs of smaller height this may fail, as the interpretations in the various $\mathbb{Q}_{[k]\backslash\{i\}}$ need not be at the same levels.  To circumvent this, we say that  (for $S\subseteq [k]$) a balanced connected digraph $\mathbb{H}$ is interpretable in $\mathbb{Q}_S$ \emph{at height $i$}, if it is interpretable in $\mathbb{Q}_S$ with a vertex of height $0$ in $\mathbb{H}$ taking the value of a height $i$ vertex of $\mathbb{Q}_S$.

\begin{lem}\label{lem:QS}
Let $\mathbb{H}$ be a connected balanced digraph.  Then $\mathbb{H}$ is interpretable in $\mathbb{Q}_S$ at height $i$ if and only if for each $j\notin S$ we have $\mathbb{H}$ interpretable in $\mathbb{Q}_{[k]\backslash\{j\}}$ at height~$i$.  For fixed $i$ and connected balanced digraph $\mathbb{H}$ of height at most $n-i$, there is a 
unique minimum set $S\subseteq [k]$ with $\mathbb{H}$ interpretable in $\mathbb{Q}_{S}$ at height $i$.
\end{lem}
\begin{proof}
The second statement follows immediately from the first, as we may successively test for interpretability (at height $i$) in $\mathbb{Q}_{[k]\backslash\{j\}}$ for $j=1,\dots,k$.  The bound $n-i$ is simply to account for the fact that if $\mathbb{H}$ has height greater than $n-i$, then it is not even interpretable in $\mathbb{Q}_{[k]}$ at height $i$. For the first statement, observe that if $j\notin S$, then there is a height-preserving homomorphism from $\mathbb{Q}_S$ onto $\mathbb{Q}_{[k]\backslash\{j\}}$ (as $S\subseteq [k]\backslash\{j\}$). So it suffices to show that if $\mathbb{H}$ is interpretable in $\mathbb{Q}_{[k]\backslash\{j\}}$ at height $i$ for each $j\notin S$ then it is interpretable in $\mathbb{Q}_{S}$ at height~$i$.  This is routine, because the single zigzag in $\mathbb{Q}_{[k]\backslash\{j\}}$ (based at height $j$) for $j\notin S$ matches the corresponding zigzag based at height $j$ in $\mathbb{Q}_{
S}$.  More formally, in the direct product $\prod_{j\notin S}\mathbb{Q}_{[k]\backslash\{j\}}$, the component connecting the tuple of initial vertices to terminal vertices maps homomorphically onto $\mathbb{Q}_
S$.
\end{proof}
\begin{defi}\label{def:gamma}
The smallest set $S\subseteq[k]$ for which a connected balanced digraph $\mathbb{H}$ is interpretable in $\Q_S$ at height $i$ is denoted by $\Gamma(\mathbb{H})^{(i)}$.  When $i$ is implicit, then we write simply $\Gamma(\mathbb{H})$.
\end{defi}

\begin{lem}\label{lem:pathassign}\hfill
\begin{enumerate}
\item $\CSP(\mathbb{Q}_{[k]})$ is solvable in logspace, even with singleton unary relations added.
\item If $\mathbb{H}$ is connected and balanced of height at most $n$, 
then for any vertex $u$ and $v$, the height of $v$ relative to that of $u$ may be computed in logspace.
\item $\CSP(\Q_{[k]\backslash\{i\}})$ is solvable in logspace for any
  $i\in \{1,\dots,k\}$, even when singleton unary relations are added.
\item For any $S\subseteq\{1,\dots,k\}$ the problem $\CSP(\Q_S)$ is
  solvable in logspace, even when singleton unary relations are added.
\item For a balanced connected digraph $\mathbb{H}$ of height at most $n$, we may test membership of numbers $j$ in the set $\Gamma(\mathbb{H})^{(i)}$ in logspace.
\item For any family of subsets $S_1,\dots,S_\ell\subseteq\{1,\dots,k\}$, the CSP over the digraph
  formed by amalgamating the family $\Q_{S_1},\dots,\Q_{S_\ell}$ at either all the initial
  points, or at all the terminal points is logspace solvable.
\end{enumerate}
\end{lem}
\begin{proof} (1) Note that $\mathbb{Q}_{[k]}$ has both a Maltsev polymorphism
and a majority, hence is solvable in logspace even when unary 
singleton relations are added \cite{Dalmau/Larose:2008:Maltsev}.
  
  (2) For each $0\leq i,j\leq n$ (the possible heights) we may test for interpretability of $\mathbb{H}$ in $\mathbb{Q}_{[k]}$ with $u$ constrained to lie at height $i$ and $v$ constrained to lie at height $j$.  As $\mathbb{H}$ is balanced of height at most $n$, at least one such instance is interpretable, and the number $j-i$ is the relative height of $v$ above $u$.
  
  (3) Note that as $\Q_{[k]\backslash\{i\}}$ is a core, we
  have $\CSP(\Q_{[k]\backslash\{i\}})$ logspace equivalent to the CSP
  over $\Q_{[k]\backslash\{i\}}$ with all unary singletons added (see \cite{lartes}).

  Given an input digraph $\H$, we first test if $\H$ is interpretable in
  $\mathbb{Q}_{[k]}$ (which verifies  that $\H$ is
  balanced, and of sufficiently small height).  Reject if NO.
  Otherwise we may assume that $\H$ is a single component.

  We successively search for an interpretation of $\H$ in $\Q_{[k]\backslash\{i\}}$ at heights $0,1,\dots,n$; in each case, item (2) shows that we have access to a suitable notion of height for the vertices of $\H$.  The remaining part of this proof concerns an attempt at interpretation into $\Q_{[k]\backslash\{i\}}$ at one particular height.  
  Any vertices of the same height $j\notin
  \{i,i+1\}$ will be identified by an interpretation, so such an interpretation exists if and only if there is no directed path of
  vertices of heights $i-1,i,i+1,i+2$.  To verify this property in logspace it suffices to enumerate all
  $4$-tuples of vertices $u_1,\dots,u_4$, check if $u_1\rightarrow u_2
  \rightarrow u_3\rightarrow u_4$, and if so, check that the height of
  $u_1$ is not $i-1$.  If it is, then reject.  Otherwise accept.

  (4) \& (5)  These follow immediately from Lemma \ref{lem:QS}, and part (3) of the present lemma.
  
  (6) We refer to a digraph formed by amalgamating paths in one of the two described fashions as a \emph{fan}.  We consider the case where the initial vertices have been amalgamated, with the case for amalgamation at terminal vertices following by symmetry.
  
   Consider some instance $\H$.  As above, we may assume that
  $\H$ is connected, balanced and is of sufficiently small height.  We may first use item (4) to test if $\H$ is interpretable in one of the individual paths $\Q_{S_1}$, $\Q_{S_2}$,\ldots.  If one of these returns a positive answer, then $\H$ is a YES instance.  Otherwise, remove all level $0$ vertices of $\H$, and successively test each individual component~$C$ of the resulting digraph for interpretability in $\Q_{S_1}$, $\Q_{S_2}$,\ldots, with an additional condition: the vertices of $C$ which were adjacent to a level 0 vertex in $\H$ must be interpreted at the level $1$ vertex of $\Q_{S_i}$ adjacent to the initial vertex.  Provided each such $C$ is interpretable in at least one of these paths in the described way, then $\H$ is interpretable in the fan (with the level $0$ vertices of $\H$ interpreted at the amalgamated initial vertices).  Otherwise, $\H$ is not interpretable in the fan and is a NO instance.
\end{proof}

Recall that we are assuming that $\G$ consists of a single component.
\subsection{Stage 1: Verification that \texorpdfstring{$\G$}{G} is balanced and a test for
  height.}\label{subsec:1} If $\G$ is not balanced of height at most $n$, then we can output some fixed NO instance and the algorithm is finished.  The logspace test for this property is Lemma \ref{lem:pathassign} part~(1).  From this point on, we will assume that $\G$ is balanced, of height at most $n$ and consisting of a single component.

\subsection{Stage 2: \texorpdfstring{$\G$}{G} has height less than \texorpdfstring{$n$}{n}.} \label{subsec:2} If $\G$ has height strictly less than $n$ then any possible interpretation of $\G$ in $\mathcal{D}(\A)$ must either
interpret $\G$ within some single path $\Q_S$ connecting $A$ to $R$ in
$\mathcal{D}(\A)$, or at some fan of such paths emanating from some
vertex in $A$ or some vertex in $R$.  There is a constant
number of such subgraphs of $\mathcal{D}(\A)$, and we may use Lemma
\ref{lem:pathassign}(6) for each one.  If $\G$ is not in $\CSP(\mathcal{D}(\A))$ then output some fixed NO instance of $\CSP(\A)$; otherwise output some fixed YES instance of $\CSP(\A)$.

For the remainder of the algorithm we will assume that $\G$ is a balanced digraph of height $n$ and consists of a single component.

\subsection{Stage 3: \texorpdfstring{$\G$}{G} has height \texorpdfstring{$n$}{n}.}\label{subsec:3} We will eventually output a
structure $\B$ with the property that $\G$ is a YES instance of
$\CSP(\mathcal{D}(\A))$ if and only if $\B$ is a YES instance of
$\CSP(\A)$.  The construction of $\B$ is a little technical, so we initially (step 3A) describe the construction of an object $\B'$ and then subsequently (step 3B) describe the construction of $\B$ from the object $\B'$.  The object $\B'$ is simply a list of information that is more easily used to output $\B$.

For the remainder of the argument, an \emph{internal component} of
$\G$ means a connected component of the induced subgraph of $\G$
obtained by removing all vertices of height $0$ and $n$.  Note that we
have already described that testing for height can be done in
logspace.  A \emph{base vertex} for such a component $C$ is a vertex
at height $0$ that is adjacent to $C$ and the set of all base vertices of $C$ is denoted $\base(C)$.  A \emph{top vertex} for $C$
is a level $n$ vertex adjacent to $C$  and the set of all base vertices of $C$ is denoted $\topp(C)$.  Note that an internal
component may have none, one, or more than one base vertices, and similarly
for top vertices.  Every internal component must have at least one of
a base vertex or a top vertex however, because $\G$ consists of a single component and has height $n$.

Let $C$ be an internal component.  In a satisfying interpretation
of $\G$ in $\mathcal{D}(\A)$, the component $C$ must  be
satisfied within some single connecting path (of the form $\Q_S$ for some $S\subseteq [k]$), with the vertices in $\base(C)$ (or $\topp(C)$) being interpreted adjacent to the
initial point of the path (or adjacent to the terminal point of the
path, respectively).   Recall (Definition \ref{def:gamma}) that $\Gamma(C)$ denotes the smallest set $S\subseteq[k]$ for which $C$ is interpretable in $\Q_S$.  By Lemma \ref{lem:pathassign}(5) we can, in logspace, verify membership of numbers up to $k$ in the set $\Gamma(C)$.  (Note that we omit the superscript ``$^{(i)}$'' in the notation, as there is no ambiguity as to what height $C$ is to be satisfied at: it is either $i=1$, or dually, at $i=n-1-\hgt C$, where $\hgt C$ denotes the height of $C$.)  These internal components  are in essence encoding positions of base level
vertices in hyperedges of the structure $\B$ in construction. Lemma
\ref{lem:pathassign}(5) supplies, in logspace, the positions which are being
asserted as ``filled'' by a given internal component $C$.   If $\mathbb{G}$ itself is the path $\mathbb{Q}_I$ for example, then the single
internal component $C$ has $\Gamma(C)=I$.

\begin{rem}\label{rem:B}
When  $\B$ is constructed, its vertices will consist of representatives of equivalence classes of a set $X$, whose members consist of the height $0$ vertices of $\G$ along with some other vertices.  The extra vertices will be added to account for information such as the fact that some internal components have no base vertices, while the equivalence relation accounts for  information such as that some base vertices are necessarily identified under any possible homomorphism from $\G$ to $\mathcal{D}(\A)$: for an internal component $C$ for example, all elements of $\base(C)$ must be identified.
\end{rem}

\subsection{Stage 3A: The object \texorpdfstring{$\B'$}{B'}.}\label{subsec:3A}
We define $\B'$ then show how it can be constructed in logspace.  For $i=0,\dots,n$, let $G_i$ denote the set of vertices of $\G$ at height $i$.  We first describe some new vertices that will be constructed.
\begin{enumerate}[label=\({\greek*}]
\item If $C$ is an internal component with $\topp(C)=\varnothing$, then for each $b\in \base(C)$ and $i\notin \Gamma(C)$ we introduce a new vertex $x_{C,b,i}$.
\item  If $C$ is an internal component with $\base(C)=\varnothing$, then for each $e\in \topp(C)$ and $i\in \Gamma(C)$ we introduce a new vertex $x_{C,e,i}$.
\item If $e\in G_n$ and $i\in\{1,\dots,k\}$ are such that no internal component $C$ with $e\in \topp(C)$ has $i\in \Gamma(C)$ then we introduce a new vertex $x_{e,i}$.
\end{enumerate}
Let $X$ denote the union of $G_0$ with all of the new the vertices just introduced.  In stage 3B, the vertices of the final output structure $\B$ will constructed as equivalence classes of the elements of $X$.
\begin{rem}\label{rem:B'} 
Recall that under any possible homomorphism from $\G$ to $\mathcal{D}(\A)$, the internal components of $\G$ must map to connecting paths in $\mathcal{D}(\A)$.  Thus if $C$ is an internal component, $b\in \base(C)$ and $i\in\Gamma(C)$, then in the construction of $\B$ we will require a copy of $b$ lying at position $i$ of some hyperedge.  But this hyperedge of $\B$ must have vertices at all $k$ positions, and the purpose of the vertices $x_{C,b,i}$ is to fill any such positions not provided by existing elements of $G_0$.   

Similarly, if some internal component $C$ with $e\in \topp(C)$ and $i\in \Gamma(C)$ has no base vertices, then $x_{C,e,i}$ is added to play the role of a vertex that lies at position $i$ of the hyperedge in $\B$ that will correspond to $e$.  
Vertices $x_{e,i}$ are added when there is no internal component at all encoding the $i$th position of $e$.
\end{rem}
Before we give a concrete description of $\B'$ we need some further notation.  If $e\in G_n$ and $i\in\{1,\dots,k\}$, let $S(e,i)$ denote the set of internal components $C$ of $\G$ for which $\base(C)\neq \varnothing$, $e\in\topp(C)$ and $i\in\Gamma(C)$.  Let $T(e,i)$ denote the set of internal components $C$ of $\G$ for which $\base(C)= \varnothing$, $e\in\topp(C)$ and $i\in\Gamma(C)$. 

Now we may describe the object $\B'$, which itself consists of a list of four different types of object.  The first and second relate to the hyperedges we will finally output for $\B$ in stage 3A.  The third and fourth record some information that will be used to determine when elements of $X$ need to be identified in the construction of $\B$.
\begin{enumerate}[label=(\Roman*)]
\item[(I)] For each vertex $e\in G_n$ we include the object $(e,V_1,\dots,V_k)$ where for each $i=1,\dots,k$ the set $V_i$ is defined as follows:
\begin{enumerate}[label=(roman*)]
\item[(i)] if $S(e,i)\cup T(e,i)\neq \varnothing$ then
\[
V_i:=\left(\bigcup \{\base(C)\mid C\in S(e,i)\}\right)\cup \{x_{C,e,i}\mid C\in T(e,i)\};
\]
\item[(ii)] if $S(e,i)\cup T(e,i)= \varnothing$ then
\[
V_i:=\{x_{e,i}\};
\]
\end{enumerate}
\item[(II)] For each vertex $b\in G_0$ and each internal component $C$ with $b\in \base(C)$ and $\topp(C)=\varnothing$, we include the object $(V_1,\dots,V_k)$, where
\[
V_i:=\begin{cases}
\{b\}&\text{ if }i\in\Gamma(C)\\
\{x_{C,b,i}\}&\text{ otherwise}
\end{cases}
\]
\item[(III)] An edge relation on $G_n$, where vertex $e\in G_n$ is adjacent to vertex $f\in G_n$ if there exists an internal component $C$ with $e,f\in \topp(C)$.
\item[(IV)] An edge relation on $G_0$, where vertex $b\in G_0$ is adjacent to vertex $c\in B_0$ if there exists an internal component $C$ with $b,c\in \base(C)$.
\end{enumerate}
\begin{lem}\label{lem:B'}
$\B'$ can be constructed in logspace from $\G$.
\end{lem}
\begin{proof}
To construct the objects of type (I) it suffices, for each $e\in G_n$ and $i=1,\dots,k$, to explain how to construct the set $V_i$.  Thus, given~$e,i$ we need to identify internal components $C$ for which $i\in\Gamma(C)$ and $e\in \topp(C)$ and then perform some checks on this $C$.  
\begin{enumerate}[label=(I.\roman*)]
\item[(I.i)] Systematically search through all
  vertices of $\G$ until some $u$ is found to be undirected-reachable
  from $e$ amongst vertices not at height $0$ or $n$.  To avoid
   duplication, also check that $u$ does not lie in
  the same internal component as some earlier vertex, ignore this $u$ if it does.  Otherwise, $u$ is the first vertex in some internal component $C_u$ for which $e\in \topp(C_u)$.  
  \item[(I.ii)] The component $C_u$ may be constructed in logspace using undirected reachability checks (using the underlying graph on the set of vertices of heights not equal to $0$ or $n$).  This component $C_u$ can be written to a query tape for verifying $i\in \Gamma(C_u)$, which is logspace by Lemma~\ref{lem:pathassign}(5).  
  \item[(I.iii)] If $i\in \Gamma(C_u)$ and $\base(C_u)$ is nonempty, then include all vertices of $\base(C_u)$ in $V_i$.  If $i\in \Gamma(C_u)$ and $\base(C_u)=\varnothing$, then include $x_{C,e,i}$ in $V_i$.
  \item[(I.iv)] If no internal component $C$ has $i\in\Gamma(C)$ and $e\in \topp(C)$, then $V_i$ is $\{x_{e,i}\}$.\medskip
\end{enumerate}

\noindent Now for objects of type (II).  Given $b\in G_0$ we need to decide if there is an internal component $C$ with $b\in\base(C)$ and $\topp(C)=\varnothing$. 
\begin{enumerate}[label=(II.\roman*)]
\item[(II.i)] Systematically search through all
  vertices of $\G$ until some $u$ is found to be undirected-reachable
  from $b$ amongst vertices not at height $0$ or $n$.  To avoid
   duplication, also check that $u$ does not lie in
  the same internal component as some earlier vertex, ignore this $u$ if it does.  Otherwise, $u$ is the first vertex in some internal component $C_u$ for which $b\in \base(C_u)$.  If $\topp(C_u)\neq \varnothing$ then ignore $u$ and continue the search.
  \item[(II.ii)] If $\topp(C_u)= \varnothing$, then we output an object of type (II) for $b$.  For each $i\in1,\dots,k$ we construct $V_i$ as follows.
  \begin{itemize}[label=$-$]
  \item Test if $i\in\Gamma(C_u)$ using Lemma~\ref{lem:pathassign}(5).  If $i\in\Gamma(C_u)$, write $\{b\}$ for $V_i$.  
  Otherwise write $\{x_{C_{u},b,i}\}$.
  \end{itemize}
  \end{enumerate}
  If the search in (II.i) returns no internal components $C$ with $b\in \base(C)$ and $\topp(C)=\varnothing$ then there is no object of type (II) for this $b$. 
  
  Next we construct objects of type (III).  This involves, for each pair of vertices $e,f\in G_n$, searching for undirected reachability in amongst the underlying graph on all vertices not in $G_0$.  This is an undirected reachability check so can be performed in logspace.  The pair $(e,f)$ is output when such an edge is found.  The construction of objects of type (IV) is almost identical, but using undirected reachability amongst vertices of height less than $n$.
\end{proof}
\subsection{Stage 3B: construction of \texorpdfstring{$\B$}{B}.}\label{subsec:3B}
We now need to construct $\B$ from the object~$\B'$. The vertices of $\B$ will consist of representatives of the blocks of some quotient of $X$.  

Define $\sim$ to be the smallest equivalence relation on $X$ satisfying the following conditions.
\begin{enumerate}
\item for each object $(V_1,\dots,V_k)$ of type (II) or $(e,V_1,\dots,V_k)$ of type (I) and each object $(W_1,\dots,W_k)$ or $(f,W_1,\dots,W_k)$, if there is $i$ such that $V_i\cap W_i\neq\varnothing$, then $V_i\times W_i\subseteq{\sim}$.
\item the edge relation of type (IV) is contained in $\sim$.
\item Assume one of the following holds: $(e,f)$ or $(f,e)$ is an edge of type (III).  Then if $(e,V_1,\dots,V_k)$ and $(f,W_1,\dots,W_k)$ are objects of type (I), then for each $i=1,\dots,k$ we have $V_i\times W_i\subseteq {\sim}$.
\end{enumerate}
For any $x\in X$, we write $\hat{x}$ to denote the lexicographically earliest member of the $\sim$-class containing $x$.  If $V\subseteq x/{\sim}$, then we also write $\hat{V}$ to denote $\hat{x}$.  Note that for every object $(V_1,\dots,V_k)$ of type (II) or $(e,V_1,\dots,V_k)$ of type (I), condition (1) guarantees that the set $V_i$ lies completely within a block of $\sim$.  Hence $\hat{V}_i$ is well defined.    

\begin{rem}\label{rem:sim}
If $a,b\in G_0$ have $a\sim b$, then any possible homomorphism from $\G$ to $\mathcal{D}(\A)$ would identify $a$ and $b$.  
\end{rem}
\begin{proof}
This is easy to verify if $a$ and $b$ are $\sim$-related because of one of the generating properties \up(1\up{),} \up(2\up{),} \up(3\up{)}.  The general case then follows immediately from the fact that the kernel of any homomorphism is an equivalence relation and $\sim$ is the smallest equivalence relation containing the pairs defined in (1), (2), (3).
\end{proof}
Now we may define $\B$.
The vertices of $\B$ will be $\{\hat{x}\mid x\in X\}$.   
The hyperedges of $\B$ consist of $k$-tuples $(\hat{V}_1,\dots,\hat{V}_k)$, for each object $(V_1,\dots,V_k)$ of type (I) or $(e,V_1,\dots,V_k)$ of type (II).

\begin{lem}\label{lem:Blogspace}
$\B$ can be constructed in logspace from $\B'$.
\end{lem}
\begin{proof}
It suffices to show how to output the hyperedges for $\B$.  There is no harm in allowing repeats in the list of hyperedges; with a further logspace process, these could be eliminated, or an adjacency matrix could be output.  

The equivalence relation $\sim$ used in the definition of $\B$ is the smallest equivalence relation containing the pairs described in items (1), (2) and (3).  Each of the kinds of pairs defined in items (1), (2) and (3) can be verified in logspace, so given any $x\in X$ we may output $\hat{x}$ in logspace; see Subsection \ref{subsec:routines}.  Thus given $V\subseteq X$, where $V$ is a subset of a block of $\sim$, we may construct $\hat{V}$ by taking the first element $x$ of $V$ and outputting $\hat{x}$.

Thus to output the hyperedges of $\B$ in logspace, perform the following.  For each object $(e,V_1,\dots,V_k)$ of type (I) or $(V_1,\dots,V_k)$ of type (II), ignore the first entry $e$ if applicable, then output $(\hat{V}_1,\dots,\hat{V}_k)$.
\end{proof}

For $e\in G_n$, let $\hat{e}$ denote the tuple $(\hat{V}_1,\dots,\hat{V}_k)$, where $(e,V_1,\dots,V_k)$ is the unique object of type (I) associated with $e$.  
\begin{rem}\label{rem:homom}
There is a homomorphism from $\G$ into $\mathcal{D}(\B)$, mapping $b\in G_0$ to $\hat{b}$ in $\mathcal{D}(\B)$, and  $e\in G_n$ to the hyperedge vertex $\hat{e}$.
\end{rem}
\begin{proof}
This comes down to the fact that an internal component $C$ of $\G$ maps onto $\Q_{S}$ for any $S\supseteq \Gamma(C)$ and the fact that the components $C$ give rise to the hyperedges of $\B$ with $i\in \Gamma(C)$ ensuring that any base vertices of $C$ appear at position $i$ of a corresponding hyperedge of $\B$; see Remark \ref{rem:B'}.  

If $C$ has $b\in \base(C)$ and $e\in \topp(C)$, then $\hat{b}$ connects to $\hat{e}$ in $\mathcal{D}(\B)$ via a connecting path $\Q_S$ with $S\supseteq \Gamma(C)$.  The component $C$ is mapped onto this connecting path.   

If $C$ has no base, but has $e\in \topp(C)$, then for each $i\in \Gamma(C)$ we added a vertex $x_{C,e,i}$ (or $x_{C,e}$ if $\Gamma(C)=\varnothing$) in $\B'$ and the vertex $\hat{x}_{C,e,i}$ connects to $\hat{e}$ in $\mathcal{D}(\B)$ via a connecting path with $\Q_S$ with $S\supseteq \Gamma(C)$.  The component $C$ is mapped onto this connecting path.  

If $C$ has no top, but has $b\in \base(C)$, then we created a object of type (II) of $\B'$ in which $b$ appeared in each position $i\in \Gamma(C)$.  After applying the $\sim$ relation this object yields a hyperedge at which $\hat{b}$ lies in at least all positions in $\Gamma(C)$.  So the component $C$ maps to the corresponding connecting path from $\hat{b}$ in $\mathcal{D}(\B)$.
\end{proof}

\subsection{Proof of Theorem \ref{thm:logspace}}
By Lemmata \ref{lem:B'} and \ref{lem:Blogspace} it suffices to show that there is a homomorphism from $\G$ to $\mathcal{D}(\A)$ if and only if there is a homomorphism from $\B$ to~$\A$.

First, any homomorphism $\phi$ from $\B$
into $\A$ extends to a homomorphism from $\mathcal{D}(\B)$ into $\mathcal{D}(\A)$ (similarly as in the proof of Lemma \ref{lem:end}), which in turn gives a homomorphism $\Phi$ from $\G$ to  $\mathcal{D}(\A)$ via Remark \ref{rem:homom}.

For the converse, let $\Phi:\G\to\mathcal{D}(\A)$ be a homomorphism.  We wish to define a homomorphism $\phi$ from $\B$ to $\A$.  If $b\in B$ is such that there exists $g\in G_0$ with $g\sim b$, then we define $\phi(b):=\Phi(g)$.  This is well defined by Remark \ref{rem:sim}.  If $b$ is not $\sim$-related to any vertex of $G_0$, then $b$ is $\sim$-related to one of the new vertices added in the construction of $\B'$: such a vertex corresponds to a position $i$ in a hyperedge vertex $e$ of $\mathcal{D}(\A)$ that is either in the image of $\Phi(\G)$, or at the top of some connecting path that intersects $\Phi(\G)$ nontrivially; see Remark \ref{rem:B'}.  In this case, define $\phi(b)$ to be the vertex of $\A$ lying at the $i$th position of $e$.  This function $\phi$ preserves the hyperedges of $\B$ because each internal component $C$ of $\G$ maps onto a connecting path $\Q_S$ of $\mathcal{D}(\A)$ for which $\Gamma(C)\subseteq S$.

\subsection{An example}\label{subsec:example}
The following diagram depicts a reasonably general instance $\G$ of
$\CSP(\mathcal{D}(\A))$ in the case that $\A$ itself is a digraph, so
that $k=2$.  We are considering stage 3, so that $\G$ is a single
connected digraph of height $4$.  The vertices at height $0$ are
$b_1,\dots,b_6$, and the vertices at height $4$ are $e_1,\dots,e_4$.
The shaded regions depict internal components: each is labelled by a
subset of $\{1,2\}$, depicting $\Gamma(C)$.
\begin{center}
  \includegraphics[scale=0.9]{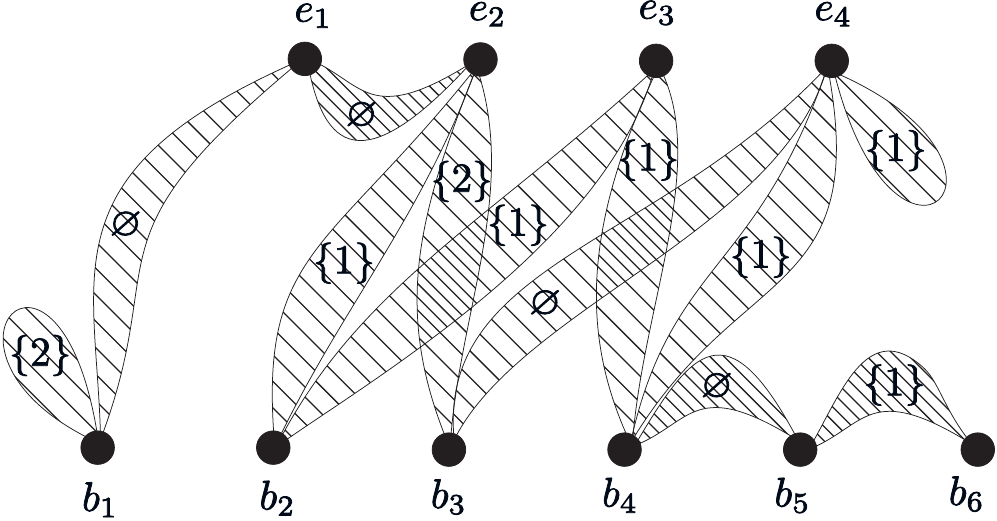}
\end{center}

Let us examine how Stage 3A proceeds.  
For objects of type (I) we obtain
\begin{align*}
&(e_1,\{x_1\},\{x_2\})\qquad&\\
&(e_2,\{b_2\},\{b_3\})\qquad&\\
&(e_3,\{b_2,b_4\},\{x_3\})\qquad&\\
&(e_4,\{b_4,x_4\},\{x_5\})\qquad&\\
\end{align*}
with $x_1,x_2,x_3,x_5$ being new vertices of kind ($\gamma$), and $x_4$ of kind ($\beta$).

For objects of type (II) we obtain 
\begin{align*}
&(\{x_6\},\{b_1\})\qquad&\text{(from $b_1$)}\\
&(\{x_7\},\{x_8\})\qquad&\text{(from $b_4$)}\\
&(\{x_9\},\{x_{10}\})\qquad&\text{(from $b_5$)}\\
&(\{b_5\},\{x_{11}\})\qquad&\text{(from $b_5$)}\\
&(\{b_6\},\{x_{12}\})\qquad&\text{(from $b_6$)}
\end{align*}
with the vertices $x_6,\dots,x_{12}$ of kind ($\alpha$).

The only object of type (III) is $(e_1,e_2)$, while $(b_4,b_5),(b_5,b_6)$ are the only two objects of type (IV).

For the construction of $\B$, the relation $\sim$ has two nontrivial blocks
\[
\{b_2,b_4,b_5,b_6,x_1,x_4\},\{b_3,x_2\}.
\]
Using the ordering $b_1<\dots<b_6<x_1<\dots<x_{12}$, the hyperedges of  $\B$ are
\[
(b_2,b_3),(b_2,b_3),(b_2,x_3),(b_2,x_5),(x_6,b_1),(x_7,x_8),(x_9,x_{10}),(b_2,x_{11}),(b_2,x_{12})
\]
(where $(b_2,b_3),(b_2,b_3)$ is listed twice only to reflect the fact that the logspace construction we gave in Lemma \ref{lem:Blogspace} would output this edge twice).

\section{Discussion}
We conclude our paper with some applications and further research directions.

\subsection*{An example}

Our construction allows us to create examples (and counterexamples) of digraph CSPs with certain desired properties, which were previously unknown or significantly harder to construct.

\begin{exa}\label{eg:fsp}
Let $\mathbb{A}$ be the structure on $\{0,1\}$ with a single $4$-ary relation 
\[
R=\{(0,0,0,1),(0,1,1,1),(1,0,1,1),(1,1,0,1)\}.
\]  
Clearly $\mathbb A$ is a core. 
Using the fact that $R=\{(w,x,y,z)\in A^4\mid w\oplus x=y\And z=1\}$ (where $\oplus$ denotes addition modulo $2$), it can be shown that the polymorphisms of $\mathbb{A}$ are the idempotent term functions of the two element group, 
and from this it follows that $\CSP(\mathbb{A})$ is solvable by the
few subpowers algorithm of \cite{IMMVW}, but is not bounded width.
Then the CSP over the digraph $\mathcal{D}(\mathbb{A})$ is also
solvable by few subpowers but is not bounded width \up(that is, is not
solvable by local consistency check\up).
\end{exa}

Prior to the announcement of this example it had been temporarily
conjectured by some researchers that solvability by the few subpowers
algorithm implied solvability by local consistency check in the case
of digraphs (this was the opening conjecture in Mar{\' o}ti's keynote
presentation at the Second International Conference on Order, Algebra
and Logics in Krakow 2011 for example).  With 78 vertices and 80
edges, Example \ref{eg:fsp} also serves as a simpler alternative to
the 368-vertex, 432-edge digraph whose CSP was shown by Atserias
in~\cite[\S4.2]{ats} to be tractable but not solvable by local
consistency check.

In \cite{JKN13}, Example \ref{eg:fsp}, Corollary \ref{cor:preserved_conditions}
and some fresh results on polymorphisms are used to construct digraph CSPs with
every possible combination of the main polymorphism properties related to decision CSPs (allowing for
Kazda's Maltsev implies majority result \cite{kaz}).

\subsection*{Which properties are preserved?}
Theorem \ref{thm:preserved_conditions} and Corollary \ref{cor:preserved_conditions} demonstrate that our reduction preserves almost all Maltsev conditions corresponding (or conjectured to be equivalent) to important algorithmic properties of decision CSPs. The linearity assumption in Theorem \ref{thm:preserved_conditions} is not limiting: Barto, Opr{\v s}al and Pinsker recently improved on the algebraic approach to the CSP by showing that, if $\mathbb A$ is a core, the complexity of $\mathrm{CSP}(\mathbb A)$ depends only on \emph{linear} idempotent Maltsev conditions satisfied by $\mathbb A$\cite{BOP15}. Still, we were not able to extend our result to include all linear idempotent Maltsev conditions (in particular, nonbalanced identities in more than two variables). \emph{Is it possible to characterize linear idempotent Maltsev conditions preserved by our construction? In particular, does it preserve precisely those which hold in the zigzag?}

In \cite{Bulatov_counting_CSPs} Bulatov established a dichtomy for counting CSPs (see also \cite{Dyer_Richerby}). The algebraic condition separating tractable (FP) problems from $\#$P-complete ones is called \emph{congruence singularity}. It is not hard to see that the structure $\mathbb A$ from Example \ref{eg:fsp} satisfies this condition and thus the corresponding counting CSP is tractable. However, congruence singularity implies congruence permutability (i.e., having a Maltsev polymorphism) which fails in $\mathcal D(\mathbb A)$. Therefore, counting CSP for $\mathcal D(\mathbb A)$ is $\#$P-complete. We conclude that our reduction does not preserve the complexity of counting. In fact, counting for $\mathcal D(\mathbb A)$ is essentially always hard. \emph{Is there a reduction of general CSPs to digraph CSPs which preserves complexity of counting?}

There are several other interesting variants or generalizations of CSPs in which algebraic conditions seem to play an important role as well. For example, infinite template CSPs (see below), valued CSPs \cite{algebraic_discrete_optimization,finite_valued_CSPs,KKR15}, or approximability of CSPs \cite{Dalmau_Krokhin_robust_satisfiability,LP_width1}. \emph{Can our construction be applied to obtain interesting results in these areas as well?}

\subsection*{Infinite template CSPs}
\noindent CSPs over infinite templates are widely encountered in artificial
intelligence; see
\cite{DBLP:journals/logcom/Hirsch97,Krokhin01reasoningabout,DBLP:books/sp/Renz02} for example.  
Efforts to obtain a mathematical foundation for understanding
the computational complexity of these problems have often involved assumptions
of model theoretic properties on the template (such as $\omega$-categoricity),
as well as the presence of polymorphisms of certain kinds;
see~\cite{DBLP:conf/dagstuhl/Bodirsky08,
DBLP:journals/logcom/BodirskyC09,MR3056109} for example.  The results of the present
article apply for such CSPs too: the proofs of Theorems
\ref{thm:preserved_conditions} and \ref{thm:logspace} did not assume finiteness
of $A$, only that
$\mathbb A$ has only finitely many relations.
\begin{rem}
  Theorem~\ref{thm:preserved_conditions} and
  Theorem~\ref{thm:logspace} extend to infinite template CSPs
  consisting of only finitely many relations. Furthermore, since
  $\mathbb A$ and $\mathcal D(\mathbb A)$ are first-order
  interdefinable, $\mathbb A$ is $\omega$-categorical if and only if
  $\mathcal D(\mathbb A)$ is $\omega$-categorical.
\end{rem}

\subsection*{Special classes of CSPs}
Hell and Rafiey \cite{ListH} showed that all tractable list homomorphism
problems over digraphs have the bounded width property, and from this it follows
that there can be no translation from general CSPs to digraph CSPs preserving
\emph{conservative} polymorphisms (the polymorphisms related to list
homomorphism problems). \emph{Find a simple restricted class of list
homomorphism problems for which there is a polymorphism-preserving translation
from general list homomorphisms to the ones in this class.}

Another class of interest are the CSPs over generalized trees.
\emph{Is there a translation from generalized trees to oriented trees that
preserves CSP tractability, or preserves polymorphism properties}?

Feder and Vardi's paper \cite{fedvar} also contains a polynomial reduction of general CSPs to CSPs over bipartite graphs. Payne and Willard announced preliminary results on a project similar to ours: to understand which Maltsev conditions are preserved by that reduction to bipartite graphs.

\subsection*{First order reductions}
The logspace reduction in Lemma
\ref{lem:reversereduction} \emph{cannot} be replaced by first order
reductions.  Indeed, it is not hard to show that $\mathcal{D}(\mathbb{A})$ is
never first order definable.  More generally though, the only first order
definable CSPs over balanced digraphs are the degenerate ones: over the single
edge, or over a single vertex and no edges (see \cite[Theorem C]{jactro}), while
deciding first order definability in general is NP-complete \cite[Theorem
6.1]{LLTb}.  Thus it seems unlikely that there is any other polynomial time
computable construction to translate general CSPs to balanced digraph CSPs (as
this would give P$=$NP).  \emph{Is there a different construction that
translates general
CSPs to \up(nonbalanced\up!\up) digraph CSPs with first order reductions in both
directions}?

\subsection*{Various reductions of CSPs to digraphs}
Feder and Vardi~\cite{fedvar} and Atserias~\cite{ats} provide
polynomial time reductions of CSPs to digraph CSPs.  We vigorously
conjecture that their reductions preserve the properties of possessing
a WNU polymorphism (and of being cores; but this is routinely
verified).  \emph{Do these or other constructions preserve the precise arity
of WNU polymorphisms}?  \emph{What other polymorphism properties are
preserved}? \emph{Do they preserve the bounded width property}?  

Translations from general CSPs to digraph CSPs need not in general be as well behaved as the
$\mathcal{D}$ construction of the present article.
 The third and fourth authors with
Kowalski~\cite{JKN13} have recently shown that a minor variation of
the $\mathcal{D}$ construction preserves
$k$-ary WNU polymorphisms (and thus the properties of being Taylor and having bounded width) 
but always fails to preserve
many other polymorphism properties (such as those witnessing strict
width, or the few subpowers property).

\section*{Acknowledgements}
\noindent The authors would like to thank Libor Barto, Marcin Kozik,
Mikl\'os Mar\'oti and Barnaby Martin for their thoughtful comments and
discussions and Andrei Krokhin for carefully reading the paper and pointing out a number of technical issues.  The authors are also indebted to two anonymous referees, whose careful reading and suggestions led to an improved presentation of Section \ref{sec:reversereduction} in particular.

%% in general the use of bibtex is encouraged
\bibliography{digraph_reduction} 
\bibliographystyle{alpha}

\end{document}